\documentclass[12pt]{article}

\usepackage{amssymb,amsmath,amsthm,amsfonts,geometry,ulem,graphicx,caption,setspace,sectsty,footmisc,caption,natbib,pdflscape,subcaption,array,hyperref,float,pgfplots,appendix,xcolor,comment,appendix,amssymb,changepage,bm,varwidth}
\usetikzlibrary{patterns}

\usepackage{makeidx}
\usepackage{titlesec}
\usepackage{appendix}
\pgfplotsset{compat=1.18}

\newtheorem{theorem}{Theorem}
\newtheorem{definition}{Definition}

\newtheorem{corollary}{Corollary}
\newtheorem{proposition}{Proposition}

\newtheorem{assumption}{Assumption}

\newtheorem{lemma}{Lemma}

\usepackage{tikz,tikzsymbols,graphicx,enumerate}

\newcommand{\ep}{\epsilon}
\newcommand{\E}{\mathbb E}

\definecolor{red1}{rgb}{.7,0,0}
\definecolor{red2}{rgb}{.85,.13,0}
\definecolor{red3}{rgb}{1,.3,.3}
\definecolor{red4}{rgb}{1,.4,.6}
\definecolor{red5}{rgb}{1,.6,.8}
\definecolor{blue1}{rgb}{0,0,.7}
\definecolor{blue2}{rgb}{0,.13,.85}
\definecolor{blue3}{rgb}{.3,.3,1}
\definecolor{blue4}{rgb}{.4,.6,1}
\definecolor{blue5}{rgb}{.7,.9,1}
\definecolor{gold}{rgb}{1,.7,.2}

\DeclareMathOperator*{\argmax}{arg\,max}

\DeclareMathOperator*{\argsup}{arg\,sup}

\title{Scoring Auctions with Coarse Beliefs}
\author{Joseph Feffer\thanks{Stanford Graduate School of Business. Email: feffer@stanford.edu. I am grateful to Mohammad Akbarpour, Benjamin Brooks, Songzi Du, Benjamin Golub, Klajdi Hoxha, Mike Ostrovsky, Andrzej Skrzypacz, Eric Tang, and Filip Tokarski and seminar participants at Stanford University and Microsoft Research for helpful comments and suggestions. }}
\date{\today}
\begin{document}

\begin{titlepage}
\maketitle
\begin{abstract}
This paper studies a simplicity notion in a mechanism design setting in which agents do not necessarily share a common prior. I develop a model in which agents participate in a prior-free game of (coarse) information acquisition followed by an auction. After acquiring information, the agents have uncertainty about the environment in which they play and about their opponents' higher-order beliefs. A mechanism admits a coarse beliefs equilibrium if agents can play best responses even with this uncertainty. Focusing on multidimensional scoring auctions, I fully characterize a property that allows an auction format to admit coarse beliefs equilibria. The main result classifies auctions into two sets: those in which agents learn relatively little about their setting versus those in which they must fully learn a type distribution to form equilibrium strategies. I then find a simple, primitive condition on the auction's rules to distinguish between these two classes. I then use the condition to categorize real-world scoring auctions by their strategic simplicity.

\end{abstract}
\end{titlepage}

\section{Introduction}
How do agents form beliefs in economic settings? One answer is that they learn passively; participants observe enough data about past gameplay for their beliefs about a state and about others to converge to a common prior (\cite{morris95}). Alternatively, agents could actively acquire information about the setting in which they play. This type of learning process may seem more realistic, especially when a game is played infrequently or when the environment changes rapidly. In such a setting, it seems unlikely that agents could learn everything about their environment. They may only be able to acquire limited or coarse information about their environment and fail to form fully precise first-order beliefs. In addition, it may be impossible to discover what opponents have learned and thus agents' higher-order beliefs may not be fully specified. In such cases, the common prior assumption may not be innocuous. Given such restrictions on the precision of participants' beliefs, many mechanisms will not be easily playable and game-theoretic solution concepts will not be implementable; agents will not have beliefs that are well-defined enough to identify an action as optimal. Thus, for a designer valuing the simplicity of agents' strategies, it may be important to distinguish between mechanisms that agents can and cannot play with a coarse learning process. 

Public procurement auctions represent an ideal setting in which to study strategizing under such restrictions. First, these auctions are almost always conducted in pay-as-bid formats that do not admit dominant-strategy or ex post equilibria. Thus, participants' strategies must depend nontrivially on their beliefs. Second, auctions for large government contracts occur too infrequently for beliefs to converge passively. However, the firms are still able to form beliefs; in practice, they often hire market researchers and auction consultants to help them acquire relevant information and bid intelligently. Third, auctioneers in public market design settings often care about strategic simplicity for participants.\footnote{A 2021 EU directive listed transparency and fairness as core principles for auction procedure. Strategically simple mechanisms are more transparent because the mapping from a firm's fundamentals to their bidding strategy is more clear. They are fairer because they require less sophistication and less market research to form best-response strategies. See \citet{eu} for more information.} Lastly, governments often care about non-price features of proposed projects, turning a simple auction setting into a complex multidimensional one. This increase in dimensionality poses additional problems for the designer in choosing a mechanism. Standard optimality criteria often yield mechanisms that are too difficult to characterize or too complicated to implement in practice. Thus, other criteria must be used to select a mechanism. The simplicity notion studied in this paper may be one useful criterion. 

Specifically, I model a scenario in which a buyer (female) seeks to obtain a good from one of two potential sellers (male)\footnote{The restriction to $2$ participants is mostly for ease of exposition. I consider auctions with a general number of participants in the discussion.}. The good can be provided at
heterogeneous quality levels, so the buyer has preferences over both the quality of the good she obtains as well as the price at which she procures it. She wants to pick a mechanism from a class of auctions called first-score auctions. These are the generalizations of first-price auctions to this procurement setting and are commonly used in practice. Before the auction, each seller receives private information about his own cost structure. However, he has no information about the state of the industry, which is the distribution from which his opponent’s type is drawn. He is, however, able to endogeneously learn coarse information about this type distribution. Each seller selects a finite number of generalized moments to investigate and learns the values of those means at the true type distribution. This process is meant to approximate the information that can be learned sufficiently well through sampling and estimation. Because information is coarse, the sellers cannot fully learn the type distribution. Thus after learning, each seller has uncertainty about the environment in which he plays and about his opponents' beliefs. I characterize the first-score auction formats in which both of these forms of uncertainty are strategically irrelevant. In the equilibria of such mechanisms called coarse beliefs equilibria (CBE), there is a bid in the auction that is a best response at all possible resolutions of this uncertainty. A CBE is formally defined through a choice of an information acquisition strategy for each player and a choice of action to play in the auction that may depend on what was learned. In equilibrium, a player will not know what his opponent has learned; he will simply know that his opponent is collecting and using information intelligently.  

\textbf{Results.} To characterize CBE, I study equilibria in which strategies can be formed with coarse information. In my first main result, I show this can happen if and only if the mechanism satisfies a property called \textit{equilibrium invariance.} This property holds when, in equilibrium, the auction's ex post allocation function is \textit{invariant} to the type distribution. That is, if type $A$ is given the contract over type $B$ at one distribution, then $A$ will be given the contract over $B$ at every possible type distribution. \textit{Equilibrium invariance} is a natural property that is satisfied in canonical auction formats---like first-price, second-price, and all-pay auctions---when agents' private information is one-dimensional. Next, I show that, while multidimensional settings generally do not admit some of the nice properties of one-dimensional ones, these properties are restored exactly when equilibrium invariance is satisfied. Theorem $2$ states that equilibrium invariance holds if and only if first-score and ``second-score'' auctions satisfy a Payoff Equivalence result. In general, the multidimensional forms of first and second-price auctions do not implement the same interim allocation rules. This makes first-score auctions harder to study than their one-dimensional counterparts because we cannot understand them through their second-score counterparts. However, when equilibrium invariance is satisfied, payoff equivalence is restored and strategies and expected outcomes of an auction are easier to analyze.  

For studying CBE, the connection between first and ``second-score'' auctions has important implications. First, it tells us that when equilibrium strategies can be formed with a finite number of moments, then they can be played with just $2$ moments. Thus, depending on their rules, first-score auctions either require agents to share fully precise, common beliefs or require them to know relatively little about their industry. Second, in the latter case, it gives a nice interpretation of the information that agents need to acquire in a CBE. Third, it allows us to derive a simple condition on the rules of a first-score auction that can distinguish mechanisms that admit CBE and those that do not. This condition is an easily testable function of the auction's rules and, in particular, does not require explicit computation of an equilibrium. Applying this condition, I show that a strict subset of scoring rules used in practice satisfy my simplicity criterion. Thus, considering the simplicity of strategies could help designers by restricting the set of mechanisms they must consider implementing. 

 In the last section of the paper, I provide microfoundations connecting a general notion of a CBE to the simpler concept studied earlier in the paper.  I formally model a prior-free game with a general information acquisition process. Agents still learn about a state, which is the true distribution of opponents' payoff-relevant types. Now, there is an abstract set of signal structures from which agents can learn, called an information technology. This is the only feature of the game that is assumed to be common knowledge. Each agent privately selects a partition of the state space from this set and learns the element of this partition that contains the state. When playing in the auction, each agent must use the information they have learned to form beliefs over opponents' incentives, information acquisition strategies, and signal realizations. I formally define a CBE with respect to an information technology in the same sense as before, by requiring an optimal auction strategy for each agent at all beliefs consistent with what he has learned. While belief structures can be quite complicated, I show in Theorem \ref{thm:CBE_informal} that when a CBE exists, it can be understood through studying corresponding equilibria in common prior games. This simplifies the problem of studying CBE by shrinking the space of beliefs that need to be considered. When the information technology relates to a sampling and estimation process, Theorem \ref{thm:CBE_informal} tells us a CBE in the formal sense coincides with the solution concept studied earlier in the paper. An immediate result of Theorem \ref{thm:CBE_informal} is that, for any information technology, no type learns the distribution fully in a CBE. However, full information about the state is dispersed among all types and thus, the same outcomes are achieved as when the state is common knowledge. As famously argued by \citet{hayek}, markets solve the problem of limited and dispersed knowledge by providing incentives to individuals to report this information truthfully. Mechanisms that admit CBE perform a similar role in settings in which beliefs matter. They solve the problem of limited beliefs by providing incentives for agents to learn about their environment and then make a report dependent on both their payoff types and what they have learned. In a CBE, the auction mechanism aggregates this information and yields an equilibrium outcome. 
  
\textbf{Related literature.} This paper primarily contributes to the literature on simplicity in mechanism design. While most works in this literature (like \cite{osp}) study failures of contingent reasoning in extensive form games, my work instantiates a more general definition of simplicity advanced in \cite{dworczak-li}. These authors study settings in which the relaxation of a standard assumption of agent behavior, in general, leads to strategic confusion. They define a simple mechanism as one in which this confusion either is never generated or is easily resolvable. In this paper, the strategic confusion stems from the coarseness of agents' learning process. Of papers that study specific forms of \cite{dworczak-li}'s general simplicity, \cite{borgers-li} and \cite{agatheisi} similarly consider belief-related notions of simplicity. The first paper considers a setting in which agents have precise first-order beliefs but cannot form accurate higher-order beliefs. They thus call a mechanism simple if agents' strategies do not depend on these higher-order beliefs. The second studies a form of simplicity when agents do not fully know their values. They find that in almost all practically relevant settings, agents should learn about opponents' values to understand what they should learn about their own.   Aside from the new form of simplicity studied, this paper offers a new direction for the literature by advancing setting-specific notions of simplicity. Rather than initially considering general games that average people can play, I focus on a specific economic setting whose participants are more sophisticated. A more general approach may miss the specific features of this setting (multidimensionality) and the specific constraints that are salient for these participants (coarse learning as opposed to reasoning) that make it interesting to study. The insights derived from this approach need not be constrained to the initial setting studied, but without having one in mind, compelling notions of simplicity may go under-studied or under-appreciated. 

Of course, this paper is not the first to take issue with and relax the strength of the common prior assumption. Two primary methods of relaxing the CPA appear in the robust mechanism design and rationalizability literatures, epitomized by \citet{bergmorr} and \citet{pearce}. These papers study games that admit equilibrium strategies and iteratively undominated strategies, respectively, under all arbitrary belief structures. Viewing belief-free implementation as another demanding restriction, intermediate solution concepts between common prior and prior-free environments have been developed. One natural way to interpolate between these two regimes is to assume agents' beliefs are defined by a finite number of moments. \citet{brooksdu} and \citet{gretschko-mass} study versions of the optimal auction design problem where the auctioneer's information about agents' values is captured by moments while 
 \citet{ollarpenta} considers moment-based belief restrictions in a rationalizability context\footnote{The concept studied by \citet{ollarpenta} relates closely to the idea of $\Delta$-rationalizability defined in \citet{rationalizability}.}\footnote{\citet{wolitzky} also considers a similar moment-based relaxation and studies robustly optimal mechanisms in a trade setting.}. I take a similar approach but apply these considerations to the agents' problem. Additionally, in my model, the moments that restrict beliefs are endogenously chosen by agents rather than being exogenously determined as in prior work. Viewed as an intermediate robustness notion, a coarse beliefs equilibrium sits between Bayes-Nash and robust implementation by requiring an optimal action on some strict subset of the space of beliefs. Another way to interpolate between a common prior and belief-free setting is to assume agents can observe data from a common dataset. \cite{modibo} and \cite{annie} both study mechanism design in such settings though with different aims. The former studies the complexity of policies that can be successfully implemented when agents have finite data. The latter studies the likelihood of different outcomes when agents can update differently from shared data. Interpreting the moments in my model as resulting from sampling and estimation, my approach resembles these authors' ideas though importantly, I abstract from noise and other practical considerations stemming from finite data.

The endogeneity of agents' learning process aligns this paper with the literature on games with information acquisition. However, the treatment of this learning process, and, as a result, its interpretation notably differ in this model. I study prior-free information acquisition about the distribution of a random variable; seminal papers in this literature like \citet{matejkamckay} study prior-dependent information acquisition about a realization of that variable. This distinction mirrors that between market research, in which companies learn general facts about their environment, and opposition research, in which businesses know who their competitors are and learn specific information about them. In practice, both of these processes are relevant. Additionally, agents' main constraints on learning fully in my model are simply that this is infeasible and not that it would be too costly.   

Finally, in terms of application, this paper contributes to the  theoretical literature on the use of a type of public procurement auction called a first-score auction. This paper, first, provides a new result about strategizing in first-score auctions. The most well-known scoring auction papers such as \citet{che} and \citet{ac08} study simple, quasilinear scoring auctions and model sellers with simple cost functions. In such a simplified context, my result is fairly trivial. On the other end of the spectrum, \citet{wang-liu}, \citet{dastidar}, and \citet{n-h} study general scoring auctions with more general cost structures and multidimensional types. As a result, they are unable to describe equilibrium strategies very nicely. My model considers a middle-ground setting with general auctions rules but simple, additively separable cost structures. This allows for a tractable but nontrivial characterization of strategies. Additionally, this paper studies evaluates scoring auctions on a criterion distinct from the revenue/payoff comparison that is the focus of these other papers.

\section{Model}\label{sec:model}
\subsection{The Designer}
A buyer runs a procurement auction in which she cares about both the quality of a good she receives, $q\in [0,1]$, and the price she has to pay for it, $p$. To obtain this good, she runs a type of auction called a first-score auction.\footnote{These auctions are generally conducted in a sealed-bid static format. In the United States, this format is required by the Federal Code of Regulations. See \citet{fed-code}.} To define this, I first define a scoring rule: 

\begin{definition}
A scoring rule $\Phi$ is a smooth function from $\mathcal C:=\mathbb R_+\times [0,1]\to\mathbb R$ that maps $(p,q)\mapsto \Phi(p,q)$ and is strictly increasing in $q$ and strictly decreasing in $p$. 
\end{definition}

\begin{definition}\label{def:fsa}
A \textbf{first-score auction} is defined as follows:
\begin{enumerate}
    \item The buyer announces a \textbf{scoring rule} $\Phi$.
    \item The participants submit bids which are proposed contracts of asking price and good quality, $(p,q)$.
    \item The winner is the participant who submitted the highest-scoring bid, evaluated by $\Phi$. He is asked to supply the contract that he bid. 
\end{enumerate}
\end{definition}

First-score auctions are the generalization of first-price auctions to multidimensional settings; a scoring rule is used to map multidimensional contracts into one-dimensional values called scores and then an analogue of a first-price auction is run on the scores. These auctions are the most common allocation mechanism used in this type of procurement. Examples of scoring rules used in practice are the following:\footnote{See \citet{lundberg} and \citet{n-h} for more details on where these rules are used and of other scoring rules used in practice.} 
     \begin{enumerate}
     \item \textbf{quasilinear} scoring rules --- $\Phi(p,q)=\phi(q)-p$ for some concave function $\phi:[0,1]\to\mathbb R$. 
     \item \textbf{price-quality ratio (PQR)} scoring rule --- $\Phi(p,q)=-\frac{p}{q}$,
     \item \textbf{quality discount (QD)} scoring rules --- $\Phi(p,q)=-p(\bar Q-q)$, for some $\bar Q>1$.
  
     \end{enumerate}

 In the first two examples, the buyer's choice of scoring rule reflects a cost-benefit and cost efficiency objectives, respectively. Though the third case which has no natural interpretation as a governmental priority, it is still comprehensible to participants: sellers will be ranked on their asking price with greater discounts given to them for greater provisions of quality. 

\subsection{Participants}
There are $2$ sellers in the auction, each with private information about his cost structure.\footnote{The restriction to $2$ participants is mostly for ease of exposition. I consider auctions with a general number of participants in the discussion.} I denote sellers' types by $(m,f)\in\Theta:=[\underline m,\bar m]\times [\underline f, \bar f]\subset \mathbb R^2_{++}$. The cost of producing the good with quality level $q\in [0,1]$ takes an additively separable form, equal to
\[m\cdot q^\eta + f\]
for some convexity parameter $\eta\in [1,\infty).$ In words, when $\eta=1$, the components of a participant's type represents his marginal cost for providing quality and fixed cost for providing the lowest-quality version of the good respectively. I denote a general element of $\Theta$ by $\theta$.

Both agents' types are drawn i.i.d. according to a type distribution $g_0\in\mathcal G$ where $\mathcal G$ is the space of strictly positive density functions over $\Theta$. Formally,
\[\mathcal G=\{g\in \mathcal L^1(\Theta,\mu):g>0, \int g(\tau)~ d\tau = 1\}\] 
where $\mu$ is the Lebesgue measure over $\Theta$. The agents do not know $g_0$ nor do they have a prior over $\mathcal G$. After learning their types and before participating in the auction, agents can acquire information about the type distribution. They can each privately learn coarse information about $g_0$, which comes in the form of moment realizations: 

\begin{definition}\label{def:moment}
Let $M:\mathcal G\to \mathbb R$ be defined by
            \[M(g)=\E_{\theta_{-i}\sim g}[\zeta(\theta_{-i})]\]
            for some bounded measurable function $\zeta:\Theta\to \mathbb R$. This function $M$ is called a \textbf{moment}. Adopting standard notation, I denote the space of these moments by $\mathcal G^*$.
\end{definition}

 Agents can choose any finite number of moments in $\mathcal G^*$ and learn their values at $g_0$. This means they can fully learn any population mean that they might want to estimate. In practice, any sampling and estimation process would only be able to estimate these means noisily. I abstract from this consideration and instead assume that sampling constraints restrict agents to learning only a finite number of means.

 After this information acquisition process, there will be many distributions whose moments match the realizations that an agent observed. Thus, because he he has no prior, his first-order (and higher-order) beliefs will not be fully determined. Each agent will form arbitrary beliefs that are admissible given what he has learned. I fully detail this belief formation process in Section \ref{sec:microfound} but summarize it at present. 
 
\subsection{Solution Concept (Informal)}\label{subsec:sol_concept}
The agents participate in a game with two stages. In the first, they acquire information about the type distribution and in the second, they participate in a scoring auction. The timing of the game can be summarized by the figure below:

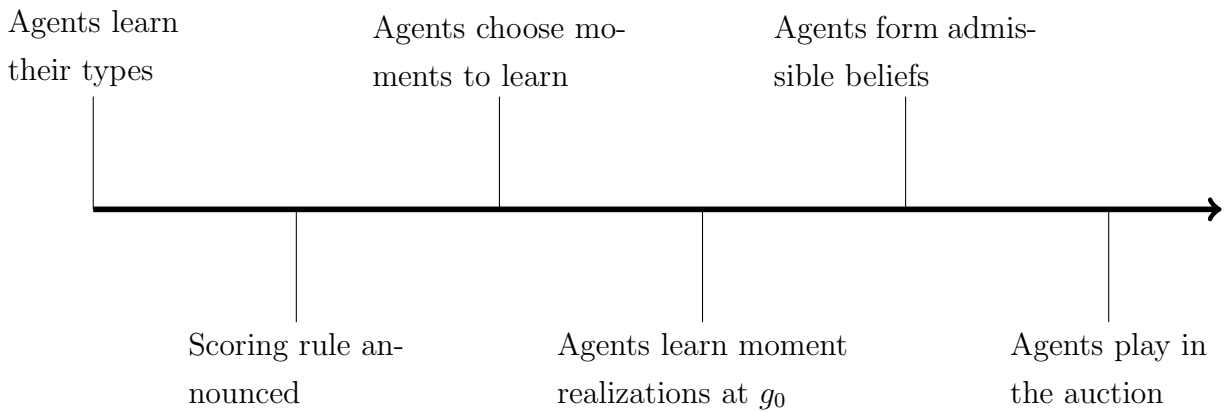
\begin{figure}[H]
            \centering
            \begin{tikzpicture}[scale=1.5]
            \draw[line width=0.75mm,->] (0,0)--(10,0);
             \draw (0,0) -- (0,1) node[above]  {\begin{varwidth}{3cm}Agents learn their types\end{varwidth}};
             \draw (9/5,0) -- (9/5,-1) node[below,align=center]  {\begin{varwidth}{4cm}Scoring rule announced\end{varwidth}};
             \draw (18/5,0) -- (18/5,1) node[above,align=center] {\begin{varwidth}{4cm} Agents choose moments to learn\end{varwidth}};
             \draw (5.4,0) -- (5.4,-1) node[below,align=center] {\begin{varwidth}{4cm}Agents learn moment realizations at $g_0$\end{varwidth}};
             \draw (7.2,0) -- (7.2,1) node[above,align=center] {\begin{varwidth}{3.5cm}Agents form admissible beliefs \end{varwidth}};
             \draw (9,0) -- (9,-1) node[below,align=center] {\begin{varwidth}{3cm}Agents play in the auction\end{varwidth}};
            \end{tikzpicture}
        \end{figure}  
        
In general, the beliefs that agents form before the auction will meaningfully affect their strategies; under one admissible belief, they should bid a certain contract and under another, they should bid differently. This is worrying because agents will have no rational method of choosing between such beliefs. In other words, the coarseness of the learning process and imprecision of beliefs generates a source of strategic confusion that the agents in this model are unable to resolve. The primary solution concept of this paper is an equilibrium concept in which this confusion is not generated.

\begin{definition}[Informal]
A mechanism has a \textbf{Coarse Beliefs Equilibrium (CBE)} if for all types, at every realization of the moments they learn, there is an action that is optimal at every admissible belief.
\end{definition}

In general, an agent will have non-Bayesian uncertainty about the true type distribution and about what his opponent has learned. In a coarse beliefs equilibrium, this uncertainty is strategically irrelevant. Coarse information about the type distribution and the understanding that his opponent is learning and bidding optimally are sufficient for forming best responses. 

The goal of this and the following section are to characterize scoring auctions that admit coarse beliefs equilibria. Most of the notation needed to formally discuss belief structures and to define coarse beliefs equilibria are not needed for this purpose. Thus, a rigorous treatment of this solution concept is provided later in the paper in Section \ref{sec:microfound}. We can characterize CBE without formally discussing beliefs due to the main result of Section \ref{sec:microfound}. This result relates CBE existence to equilibria of analogous games in which both agents share the common prior $g$: 

\noindent\textbf{Theorem \ref{thm:CBE_informal}} (Informal)\textbf{.}\textit{
A coarse beliefs equilibrium exists if and only if, for all $g\in\mathcal G$, Bayes-Nash equilibrium strategies in the common prior game at $g$ can be expressed with coarse information.}

Games with common priors are simpler to study than this two-stage game with information acquisition. Thus, we shift attention to common prior games and their equilibria for the remainder of Sections \ref{sec:model} and \ref{sec:results}.

\subsection{CBE Through Common Prior Games}
Assuming agents share a common prior $g\in\mathcal G$, I look for a strategy $(p^*,q^*):\Theta \times \mathcal G\to\mathcal C$ which is a contract to play based on a type and the type distribution such that:

\[(p^*,q^*)(\theta,g)\in \argsup_{p,q\in\mathcal C} \mathbb P_{\theta_{-i}~g}[\Phi(p,q)\ge (p^*,q^*)(\theta_{-i},g)]\cdot (p-m\cdot q^\eta-f)\] 

Note the dependence of strategies on the type distribution $g$ in this notation. In standard mechanism design settings, this distribution is fixed and so dependence of strategies on it is suppressed. We would like to consider settings in which agents can play symmetric equilibrium strategies without fully knowing the type distribution at all $g\in\mathcal G$. 

To guarantee a symmetric equilibrium exists, I place the following assumptions on the set of scoring rules the buyer considers:

\begin{assumption}[Regularity]\label{assump_scoring}
Fix a scoring rule $\Phi$. Let $u(s,\theta)$ be the indirect utility type $\theta$ receives from bidding score $s$. Let $P(s,q)$ be the price needed to obtain a score $s$ when bidding a quality $q$ under this scoring rule.
\begin{enumerate}[(a)]
    \item For all $s$, the function $\frac{\partial}{\partial m}\left(\frac{u(s,\theta)}{u_1(s,\theta)}\right)$ is continuous in $\theta$ and is strictly negative.
    \item For all $s\in \text{Im } \Phi$, we have that $P(s,q)<\infty$ for all $q\in[0,1]$. Further, we have that $P(s,q)$ is strictly convex in $q$, with smooth, strictly increasing derivative. Additionally, $P_q(s,0)<-\bar m$ and $P_q(s,1)>-\underline m$.
\end{enumerate}
\end{assumption}

These conditions, adapted from \citet{n-h}, guarantee the existence and uniqueness of a symmetric Bayes-Nash equilibrium in the game in which the type distribution $g$ is common knowledge. Assumption \ref{assump_scoring}(a) guarantees a form of single-crossing holds in this multidimensional setting and that pure-strategy equilibria will always exist as a  result. Further, it tells us the unique distribution of scores bid in equilibrium has a density function, denoted $h_g(s)$. Assumption \ref{assump_scoring}(b)  guarantees that any type's optimal contract bid is unique. 

Theorem \ref{thm:CBE_informal} tells us that a CBE exists when these optimal contract bids can be expressed with a finite number of moments:

\begin{definition}\label{def:detail_dependence}
            An agent plays an \textbf{n-moment strategy} if for all types $\theta\in\Theta$ and distributions $g\in \mathcal G$, his strategy $(p_i,q_i)(\theta,g)$ can be expressed as
            \[(p_i,q_i)(\theta,g)=\varphi_\theta( M^1_\theta(g),...,M^n_\theta(g)).\]
            for a finite sequence of moments $M^1_\theta,...M^n_\theta\in\mathcal G^*$
            and a function $\varphi_\theta:\mathbb R^{n}\to \mathbb R_+\times [0,1]$.
\end{definition}

\noindent\textbf{Corollary \ref{cor:fmit_nmoments}} (Informal)\textbf{.}\textit{
A first-score auction has a coarse beliefs equilibrium with if and only if its common prior equilibria can be implemented in $n$-moment strategies for some $n<\infty$. If this holds, we say that it can be \textbf{implemented in coarse beliefs}.}

 From this characterization, it is easy to see that popular auction formats admit coarse beliefs equilibria. For example, consider a first-price auction between $2$ participants. They have private values $v\in[0,1]$ drawn i.i.d from a distribution with density $g\in\mathcal G$, have quasilinear utility functions, and are risk-neutral expected profit maximizers. Bayes-Nash equilibrium strategies for both players are to submit bids 
     $b^*:[0,1]\times \mathcal G$ defined by \begin{equation}\label{ex:fpa_start}b^*(v,g)=\frac{\int x\cdot \mathbb I_{x\le v}~g(x)dx}{\int \mathbb I_{x\le v}~g(x)dx}.\end{equation} 
These strategies are $2$-moment strategies as the numerator and denominator on the right-hand side of equation \ref{ex:fpa_start} are both moments. Second-price and all-pay auctions in which agents possess one-dimension of private information also admit equilibria with $n$-moment strategies for finite $n$. This suggests that agents' information about their environment does not need to be too precise to form adequate best-responses. This will not hold generically in multidimensional settings. 

Note that by allowing for any finite number of moments, the definition of a CBE is quite permissive. This will turn out to be of little consequence when the condition is satisfied; if a first-score auction admits a coarse beliefs equilibrium under this definition, it will also admit one under more restrictive definitions. In the negative case, when an auction does not admit a coarse beliefs equilibrium, agents really must have full knowledge of the type distribution to be able to strategize appropriately. In the latter cases, agents need full \textit{informational precision} about the type distribution to strategize while in the former, they can do so with lower precision. 

For the sake of tractability, I make the following assumptions about strategies:
\begin{assumption}[Differentiability]\label{assump:diff}
\begin{enumerate}[(a)]
\item For all types $\theta\in\Theta,$ we have that $\varphi_\theta$ is differentiable.
\item Let $s^*=\Phi\circ (p^*,q^*):\Theta\times \mathcal G\to \mathbb R$ be the score bidding strategy in equilibrium. Then for all types $\theta\in\Theta$ and $g\in\mathcal G$, we have that $s^*(\theta,g)$ is differentiable with respect to fixed costs.
\end{enumerate}
\end{assumption}

\section{Characterization}\label{sec:results}
In this section, I characterize the first-score auction formats that admit CBE's and discuss properties of those equilibria. My main result shows there are only two classes of first-score auctions: 

\begin{theorem}\label{thm:main_theorem}
    A first-score auction either can be implemented in $2$-moment strategies or it cannot be implemented in $n$-moment strategies for any $n<\infty$.
\end{theorem}

 I will sometimes refer to auctions in the former case and their associated strategies as informationally coarse and those in the latter case as fully precise. In Section \ref{sub:fsa_strat}, I characterize equilibrium strategies in a first-score auction for general scoring rules and define a property called \textit{equilibrium invariance} that guarantees an equilibrium in informationally coarse strategies. In Section \ref{sub:diff}, I show that any mechanism that admits a CBE must satisfy equilibrium invariance. In section \ref{sub:ei_conds}, I provide a simple condition on scoring rules to determine if the corresponding first-score auctions are implementable in $2$-moment strategies.

\subsection{First-Score Strategizing}\label{sub:fsa_strat}
We will start by understanding equilibrium strategies in general first-score auctions. To discuss these strategies, I first define second-score auctions, the generalization of second-price auctions to this multidimensional setting. After this, I define an equivalence relation on the set of types, calling two types equivalent if they bid the same scores in equilibrium. With these two components, I then characterize Bayes-Nash strategies. 

In one dimension, studying second-price auctions helps us understand strategies in a first-price auction. An analogous result will hold here.

\begin{definition}
A \textbf{second-score auction} is defined as follows:
\begin{enumerate}
    \item The buyer announces a \textbf{scoring rule} $\Phi$.
    \item The participants submit bids which are proposed contracts of asking price and quality, $(p,q)$.
    \item The winner is the participant who submitted the highest-scoring bid, evaluated by $\Phi$. He is asked to supply a contract whose score equals that of the second-highest bid.
\end{enumerate}
\end{definition}

We see second-score auctions far less frequently than first-score auctions in practice. Nevertheless, these auctions remain useful as analytical tools. It is important to note that, in a second-score auction, the winner does not need to provide the second-highest scoring contract that was submitted; they can provide a contract that is optimal for them given that it attains the requisite score. This means that if two types have very different cost structures---e.g. one (the other) has a high (low) marginal cost and  has a low (high) fixed cost)---the distance between their types and the difference between their resulting bids does not penalize either of them. Because of this feature, the incentives in a second-score auction resemble those of a second-price auction. As a result, their equilibria have similar properties:

\begin{proposition}(\cite{n-h})
A second-score auction with scoring rule $\Phi$ that satisfies Assumption \ref{assump_scoring} has a unique equilibrium in weakly dominant strategies. These strategies, $(p_{BE},q_{BE}):\Theta\to \mathcal C$ represent the highest-scoring contracts a type can provide while still obtaining a weakly positive profit.
\end{proposition}

Because scoring rules are strictly increasing in price, a type will obtain exactly zero profit from the contract he bids in a second-score auction. For this reason, I refer to these bids as break-even (BE) contracts. As discussed earlier, second-score auctions and break-even contracts will be useful for describing strategies in first-score auctions, as in the one-dimensional case. However, unlike in the one-dimensional case, this result is not due to revenue equivalence. In general, revenue equivalence will not hold in this setting because the equilibria in the two auctions do not always have the same allocation rules. To formalize this, I define a notion of an equilibrium order over the set of types.

\begin{definition}
At the type distribution $g$, an equilibrium of a scoring auction \textbf{implements the order} $\succeq_g$ over $\Theta$ where 
\[s^*(\theta_1,g)\ge s^*(\theta_2,g) \Longleftrightarrow \theta_1 \succeq_g \theta_2.\]
\end{definition}

In words, one type is greater than another at a distribution $g$ if, in the equilibrium at that distribution, the former type wins the contract over the latter. The relation $\succeq_g$ is a total order as the transitivity and completeness of the order used to compare scores extend to the order $\succeq_g$. I sometimes use the phrases equilibrium structure or equilibrium order to refer to the type ordering defined by $\succeq_g$. In a second-score auction, this equilibrium order is the same for all $g\in\mathcal G$. This is generically not be the case in a first-score auction. 

This order induces an equivalence relation over the set of types that bid the same score. I denote this relation by $\sim_g$. The equivalence classes of this relation are analagous to one-dimensional ``pseudotypes'' discussed in the scoring auction literature.\footnote{ \citet{ac08} discuss pseudotypes in quasilinear scoring auctions while \citet{n-h} define pseudotypes in general first-score settings.} As in those contexts, a distribution over $1$-dimensional psuedotypes rather than that over $2$-dimensional types is all that is necessary for describing strategies.

\begin{proposition}\label{prop:fsa_strat}
Fix $g\in \mathcal G$. For all $(m,f)\in \Theta$, this type's equilibrium strategy can be represented as 
                \[(p^*,q^*)(m,f,g)=(p_{BE},q_{BE})(m,f_{(2)})\]
where
                \begin{equation}\label{eq:f2}f_{(2)}=\frac{\int z\cdot \mathbb I_{z\ge f}~ g_m([(m,z)]_g) dz}{\int \mathbb I_{z\ge f} ~g_m([(m,z)]_g) dz}.\end{equation}
Here, $g_m([(m,z)]_g)$ refers to the pushforward density of $g$ onto the set of equivalence classes when parametrized by the types with marginal cost $m$.
\end{proposition}

 Proposition \ref{prop:fsa_strat} reduces strategizing in a first-score auction into choosing a type to imitate and then simply bidding that type's break-even contract. The type to mimic can be interpreted as the expectation of the second-highest bidder's type under the pushforward distribution. This decomposition of first-score auction generalizes strategies used in one-dimensional, first-price settings. In those auctions, the Bayes-Nash strategy is to bid the break-even bid of a different type, where this break-even bid is simply that type's value. The type to imitate can be expressed as an expected conditional order statistic in Equation \ref{eq:f2}  which resembles that in Equation \ref{ex:fpa_start}. 

There are two key observations at work in proving Proposition \ref{prop:fsa_strat}. The first results from the additive separability in firms' cost structures: conditional on bidding a score $s$, agents with the same marginal cost will pick the same contract because their fixed cost does not effect their maximization problem. Mathematically, if $h_g(s)$ is the density of the score distribution in equilibrium,
\[\argmax_{(p,q)\in C(s)} h_g(s)\cdot (p-m\cdot q^\eta - f_1)=\argmax_{(p,q)\in C(s)} h_g(s)\cdot (p-m\cdot q^\eta - f_2).\]
where \[C(s)=\{(p,q)\in\mathcal C:\Phi(p,q)=s\}.\]
This immediately allows the description of a general bid as some other type's break-even bid.
Deriving an expression for this other type requires an additional observation, illustrated in Figure \ref{fig:fsa_strat}. Intuitively, if types $\theta_1$ and $\theta_2$ are equivalent, then from the strategizing perspective of one agent, it does not matter whether their opponent has type $\theta_1$ or $\theta_2$. More formally, equilibrium strategies best-respond to an equilibrium distribution of scores, $H_g(s)$. Moving mass between types in the same equivalence class, then, does not change any type's equilibrium strategies because the equilibrium score distribution remains unchanged. Thus, to study strategies, it is without loss of generality to move all mass in an equivalence class to a type that has some fixed marginal cost and then consider the resulting auction in which an agent's only private information is his fixed cost.

\begin{figure}
\centering
\begin{subfigure}{0.45\textwidth}
    \begin{tikzpicture}
    \draw[<->,thick] (0,5) node[left]{$f$}-- (0,0) -- (5,0) node[below]{$m$};
    \draw [blue2,thick,domain=215:260] plot ({5.25+6*cos(\x)}, {6+6*sin(\x)});
   \draw [black,thick,domain=215:260] plot ({5.25+6*cos(\x)}, {7+6*sin(\x)});
   \draw [red2,thick,domain=215:260] plot ({5.25+6*cos(\x)}, {8+6*sin(\x)});
   \draw [black,thick, dashed] (2,5) -- (2,-.2) node[below]{$m_0$};
   \fill[fill=blue2] (2,.95) circle (2pt);
   \fill[fill=black] (2,1.95) circle (2pt);
   \fill[fill=red2] (2,2.95) circle (2pt);
\end{tikzpicture}
\end{subfigure}
\begin{subfigure}{0.45\textwidth}
    \begin{tikzpicture}
    \draw[<->,thick] (0,5) node[left]{$f$}-- (0,0) -- (5,0) node[below]{$m$};
    \draw [blue2,opacity=0.2,domain=215:260] plot ({5.25+6*cos(\x)}, {6+6*sin(\x)});
   \draw [black,opacity=0.2,domain=215:260] plot ({5.25+6*cos(\x)}, {7+6*sin(\x)});
   \draw [red2,opacity=0.2,domain=215:260] plot ({5.25+6*cos(\x)}, {8+6*sin(\x)});
   \draw [black,thick, dashed] (2,5) -- (2,-.2) node[below]{$m_0$};
   \draw[black,line width=0.75mm,->] (2.65,4.5)--(2.15,4.5);
        \draw[black,line width = 0.75mm,->] (1.35,4.5)--(1.85,4.5);
        \fill[fill=blue2] (2,.95) circle (4pt);
   \fill[fill=black] (2,1.95) circle (4pt);
   \fill[fill=red2] (2,2.95) circle (4pt);
\end{tikzpicture}
\end{subfigure}
\caption{On the left are three equivalence classes of types in an equilibrium. To any type of agent, the distribution of density within these equivalence classes does not matter. He only cares about the aggregate density of types in that class who will bid the same score. Thus, it is without loss to consider a setting in which all of this mass is pushed to one point, as shown on the right. This then reduces the problem to a one-dimensional first-price auction.}\label{fig:fsa_strat}
\end{figure}
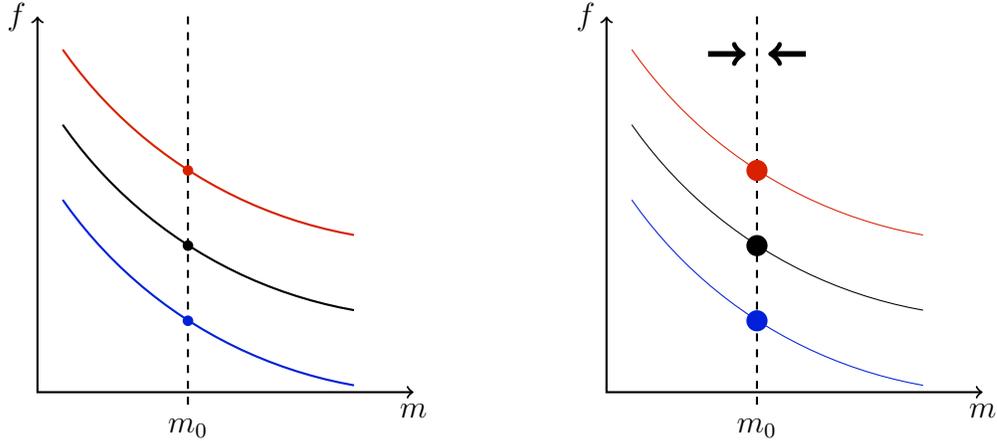

\pagebreak

Proposition \ref{prop:fsa_strat} will be useful in studying CBE generally, but it is immediately applicable in studying equilibria with $2$-moment strategies. Under Assumption \ref{assump_scoring}, $(p_{BE},q_{BE})$ is a differentiable function. Thus, if each of the integrals in Equation \ref{eq:f2} is a moment, then Proposition \ref{prop:fsa_strat} would immediately yield $2$-moment strategies. The problem is that the equivalence classes under $\sim_g$ are equilibrium objects that, in general, change with the type distribution. In the cases where they do not, only $2$ moments will be needed for strategizing. This is formalized below:

\begin{definition}\label{def:eq_inv}
A first-score auction satisfies \textbf{equilibrium invariance} if for all $\theta_1,\theta_2\in\Theta $, if $\theta_1$ is awarded the contract over $\theta_2$ in equilibrium at some distribution $g\in \mathcal G$ then $\theta_1$ is awarded the contract over $\theta_2$ at all distributions in $\mathcal G$. Mathematically, a first-score auction satisfies equilibrium invariance if for all $g_1,g_2\in\mathcal G$,
\[\theta_1\succeq_{g_1} \theta_2 \Longleftrightarrow \theta_1\succeq_{g_2} \theta_2.\]
\end{definition}

This property is naturally appealing; equilibrium invariance indicates that there is a concrete notion of one type being a ``stronger bidder than'' another. Viewing an allocation function as a choice rule, this property relates to an Independence of Irrelevant Alternatives condition in a social choice setting or a substitutability condition in a matching one because the designer's preference of one type over another does not depend on any other factors. Additionally, many common auction formats satisfy this property. When there is only one dimension of private information, there is a canonical total order over the set of types. Popular auction formats---like first-price, second-price, and all-pay auctions---implement allocation rules that are monotonic with respect to this order. This definition of monotonicity is independent of the type distribution so equilibrium invariance is satisfied. With multidimensional private information, there is no canonical total order over type spaces so this property is not readily implied by monotonicity requirements. Thus, multidimensional auctions satisfying equilibrium invariance generalize the nice structure of equilibria in one-dimensional settings. Returning to Proposition \ref{prop:fsa_strat}, equilibrium invariance implies that equivalence classes are independent of the type distribution and gives us the following corollary:

\begin{corollary}
    If an auction satisfies equilibrium invariance then it has an equilibrium in $2$-moment strategies.
\end{corollary}

\subsection{Equilibrium Structure and Finite Moments}\label{sub:diff}
The previous section concluded with a sufficiency result for an informationally coarse strategies. The main result of this section is a stronger converse, showing that equilibrium invariance is not only necessary for obtaining a precision of $2$, it is necessary in any coarse beliefs equilibrium:

\begin{proposition}\label{prop:converse}
If a first-score auction admits a coarse beliefs equilibrium, then it satisfies equilibrium invariance.
\end{proposition}

When equilibrium invariance does not hold, full information about the type distribution is needed to determine the equilibrium order. Because this order is needed to form strategies, such an auction will not admit a CBE. A proof of this argument relies on the fact that the strategy of type $\theta$ at $g$ only depends on the distribution over and actions of types in the set of types that lose to $\theta,$ per Proposition \ref{prop:fsa_strat}. If a CBE exists, then at each $g\in\mathcal G$, the moments used in strategizing should then relate closely to this set. If equilibrium invariance does not hold, then, when we consider all $g\in \mathcal G$, there are an infinite number of such sets. They cannot be described sufficiently well with only a finite number of moments and so informational precision must be high. 

To make this argument more precise, we focus on how $n$-moment strategies change as the distribution $g$ changes. First, we must develop some new terminology. We define the set of directions $\mathcal G_0$ as
\[\mathcal G_0=\{v\in \mathcal L^1(\Theta,\mu):\int v(\tau)~d\tau=0\}\]
For any $g\in\mathcal G$ and $v\in\mathcal G_0$, for small enough $\ep>0$, the distribution $g+\ep v\in\mathcal G$. Thus, it makes sense to consider changes in strategies as $g$ is perturbed in the direction $v$.  Recall that $n$-moment strategies can be describled by a differentiable function $\varphi_\theta$ and $n$ generalized moments. For each of these moments, by linearity, 
\[\frac{\partial M^i_\theta(g+\ep v)}{\partial \ep}\bigg\rvert_{\ep=0}=\lim_{\ep\to 0}\frac{\int \zeta_\theta^i(\tau) [g(\tau) + \ep v(\tau)]~d\tau - \int \zeta_\theta^i(\tau) g(\tau)~d\tau}{\ep}=\int \zeta_\theta^i(\tau) v(\tau)~d\tau.\]
Let $\sigma_\theta=\Phi\circ \varphi_\theta$ be the score bidding strategy as a function of the $n$ moments. Additionally, slightly abusing notation, let 
\[\sigma_{\theta,i}(g):=\frac{\partial \sigma_\theta(M^1_\theta(g),...,M^{i-1}_\theta(g),z,M^{i+1}_\theta(g),...,M^n_\theta(g))}{\partial z}\bigg\rvert_{z=M^i_\theta(g)}.\]
The subscript $i$ corresponds to the partial derivative of the score bid with respect to the $i$th moment. We can now express the change in the score bid as $g$ is perturbed in the direction $v$ as
\[\frac{\partial s^*(\theta,g+\ep v)}{\partial \ep}\bigg\rvert_{\ep=0}=\int \left(\sum_{i=1}^n \sigma_{\theta,i}(g) \zeta^i_{\theta}(\tau) \right)v(\tau)~d\tau.\]

This holds for any $v\in\mathcal G_0$. Thus, the derivative of an $n$-moment strategy for a type at a distribution is a single moment on the space of directions $\mathcal G_0.$ We define the function $\psi_{\theta,g}:\Theta\to\mathbb R$ to describe this derivative in the same way that the functions $\zeta_\theta^i$ describe the moments used in strategizing. That is, we define
\[\psi_{\theta,g}(\tau):=\sum_{i=1}^n \sigma_{\theta,i}(g) \zeta^i_{\theta}(\tau).\]

With this terminology, the notion of moments ``relating closely'' to the equilibrium structure is made more precise through the following lemma:

\begin{lemma}\label{lem:moment_match}
Assume an auction can be implemented with $n$-moment strategies and fix a type $\theta\in\Theta$. Then, the following must hold:
\begin{enumerate}[(a)]
\item Almost surely, $\psi_{\theta,g}(\tau)$ is constant on $\{\tau\in \Theta:\tau\succ_g \theta\}$.
\item Let $c$ be the value of $\psi_{\theta,g}(\tau)$ on this set. Then there exists a type $\theta'\prec_g \theta$ such that, almost surely, \[\psi_{\theta,g}(\tau)\not=c\text{ for }\{\tau\in\Theta:\theta' \prec_g \tau\prec_g \theta\}.\]
\end{enumerate}
\end{lemma}

\begin{figure}[t]
\centering
\begin{subfigure}{0.45\textwidth}
    \begin{tikzpicture}
   \fill[pattern color=blue2,pattern=north west lines,domain=215:260] plot (.3,.2)--plot ({5.25+6*cos(\x)}, {7.5+6*sin(\x)})--(4.2,.2)--(.3,.2);
   \fill[color=white,domain=215:260] (.3,.2)--plot ({5.25+6*cos(\x)}, {7+6*sin(\x)})--(4.2,.2)--(.3,.2);
   \fill[pattern color=red2,pattern=crosshatch dots,domain=215:260] (.3,.2)--plot ({5.25+6*cos(\x)}, {7+6*sin(\x)})--(4.2,.2)--(.3,.2);
   \draw[<->,thick] (0,5) node[left]{$f$}-- (0,0) -- (5,0) node[below]{$m$};
   \draw [black, thick,domain=215:260] plot ({5.25+6*cos(\x)}, {7+6*sin(\x)});
   \fill[fill=black] (1,2.75) circle (2pt) node[right]{$\theta$};
\end{tikzpicture}
\subcaption{The red dotted region shows the set of types that win against $\theta$ at $g$. The blue lined region shows the set of competitive types that lose against $\theta$. On the red region, $\psi_{\theta,g}(\tau)$ must be constant almost surely. On the blue region, this function cannot equal its value on the red region almost surely.}\label{fig:converse_a}
\end{subfigure}
\begin{subfigure}{0.45\textwidth}
    \begin{tikzpicture}
    \draw[<->,thick] (0,5) node[left]{$f$}-- (0,0) -- (5,0) node[below]{$m$};

   \draw [dashed,blue2, thick,domain=215:260] plot ({5.25+6*cos(\x)}, {7+6*sin(\x)});

   \draw [dashed,blue4, thick,domain=220:247] plot ({8.07+10*cos(\x)}, {9.82+10*sin(\x)});
   \draw [dashed, black, thick,domain=218:270] plot ({4.05+4.3*cos(\x)}, {5.8+4.3*sin(\x)});

   \fill[fill=black] (1,2.75) circle (2pt);

   \fill[fill=black] (4,.9) circle (2pt) node[right]{$\tau_2$};
   \fill[fill=black] (4,.4) circle (2pt) node[right]{$\tau_1$};
   \fill[fill=black] (4,1.3) circle (2pt) node[right]{$\tau_3$};

\end{tikzpicture}
\subcaption{The dashed curves are the equivalence classes $[\theta]_{g_i}$ for $i\in \{1,2,3\}$ with darker colors signifying higher $i$. Applying Lemma \ref{lem:moment_match}, we get that
$\psi_{\theta,g_1}(\tau_1)\not=\psi_{\theta,g_1}(\tau_2),$ and\\
$\psi_{\theta,g_2}(\tau_1)=\psi_{\theta,g_2}(\tau_2)\not = \psi_{\theta,g_2}(\tau_3).$}\label{fig:converse_b}
\end{subfigure}
\caption{The main ideas in the proof of Proposition \ref{prop:converse} are shown above.}\label{fig:converse}
\end{figure}
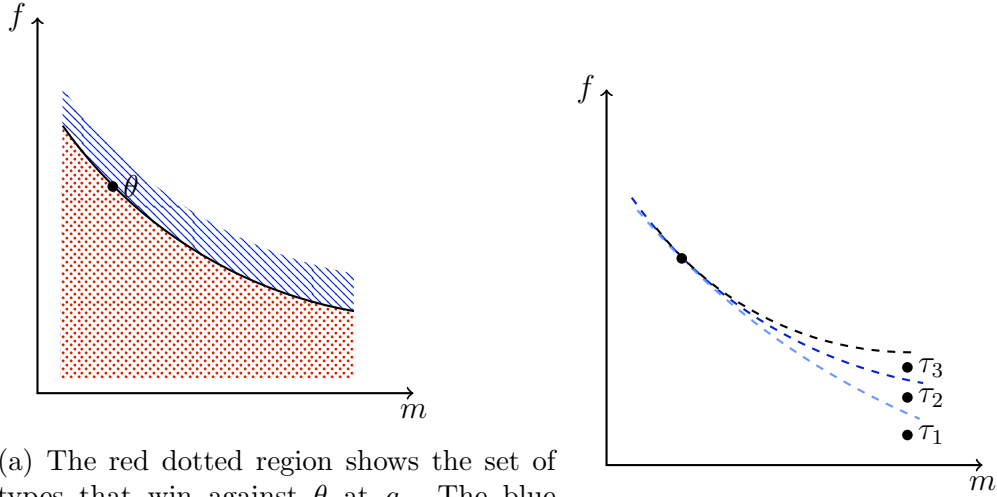
The first part of the lemma implies that if mass is moved within the set of types that win the contract over type $\theta$, then locally type $\theta$'s strategy should not change.  The second part of the lemma implies that if mass is moved from types that win the contract over $\theta$ to ``competitive types'' against which $\theta$ wins the auction, type $\theta$'s strategy should change. This is in line with standard intuition; with this type of mass movement, the conditional mean costs of the types that $\theta$ beats should decrease. Thus, $\theta$'s bid should get more competitive. Both parts of this lemma are summarized in Figure \ref{fig:converse_a}. 

With this lemma, the proof of Proposition \ref{prop:converse} is relatively simple. If an auction can be implemented in $n$-moment strategies, then as $g$ varies, the values of $\psi_{\theta,g}(\tau)$ for a fixed $\tau$ can be defined by an $n$-dimensional vector. Lemma \ref{lem:moment_match} allows us to show this cannot be the case for every $\tau\in\Theta$.
 
\subsection{Testing For CBE's}\label{sub:ei_conds}

In the previous section, I formed an equivalence between mechanisms admitting CBE's and those satisfying equilibrium invariance. The primary benefit of this was the sharp characterization of informationally coarse ($n=2$) and fully precise ($n=\infty$) mechanisms. Another way to view this result is that it translates the definition of CBE's, a statement defined on action spaces, to one defined on type spaces. While thinking about strategies in terms of action spaces leads to more interpretability, doing so leads to more analytical difficulties than considering type spaces. With equilibrium invariance, we are able to take a direct revelation approach and characterize informationally coarse auctions using standard mechanism design tools. First, I define effort to be the convex transformation of quality, $q^\eta$, that appears in participants' cost structures.  With this reparametrization, agents have linear utility structures. I then define a symmetric direct revelation mechanism $(x,y,t):\Theta^2\times \mathcal G$ where
\begin{itemize}
\item $x(\theta_1,\theta_2,g)\text{ is the probability $\theta_1$ is awarded the contract over $\theta_2$ at the distribution $g$,}$
\item $y(\theta_1,\theta_2,g)$\text{ is the effort $\theta_1$ provides conditional on beating $\theta_2$ at the distribution $g$,}\footnote{In a first-score auction, this will not depend on $\theta_2$ because of the chosen auction format. We will be using this notation to discuss second-score auctions as well so dependence on $\theta_2$ is included here.}
\item $t(\theta_1,\theta_2,g)\text{ is the transfer $\theta_1$ receives conditional on beating $\theta_2$ at the distribution $g$,}$
\end{itemize}
 I use capital letters to denote interim allocations, writing
\[Z(\theta_1,g)=\E_{\theta_{2}\sim g}[z(\theta_1,\theta_2,g)]\]
for $z\in \{x,y,t\}$. \\

In equilibrium at a specific distribution $g$, for all types $\theta=(m,f)$, we must have that the usual individual rationality and incentive compatibility constraints hold:
\begin{equation}\tag{IC}X(\theta,g)(T(\theta,g)-mY(\theta,g)-f)\ge X(\hat \theta,g)(T(\hat \theta,g)-mY(\hat \theta,g)-f)~\forall \hat \theta\in \Theta,\end{equation}

\begin{equation}\tag{IR}X(\theta,g)(T(\theta,g)-mY(\theta,g)-f)\ge 0.\end{equation}

Equilibrium invariance is equivalent to the ex post allocation rule being independent of the type distribution, or $x(\theta_1,\theta_2,g)=: x(\theta_1,\theta_2)$ for all $g\in\mathcal G$. This restricts the set of ex post allocation rules that can be implemented in an informationally coarse auction; some allocation rules can be implemented at some distribution in $\mathcal G$ but will not yield an IC mechanism at some other distribution. Considering this stronger form of incentive compatibility yields Theorem \ref{thm:rev_equiv}.

\begin{theorem}\label{thm:rev_equiv}
A first-score auction with scoring rule $\Phi$ is informationally coarse if and only if it implements the same interim outcomes as the second-score auction with the same scoring rule at all distributions $g\in\mathcal G$.
\end{theorem}

With one dimension of private information, the Revenue Equivalence Theorem gives a quick derivation of the first-price auction's $2$-moment strategies. In this setting, when a first-score auction has $2$-moment strategies, the appropriate generalization of the Revenue Equivalence Theorem provides the same guarantee. Additionally, Theorem \ref{thm:rev_equiv} gives more interpretability of the moments learned by the agents in informationally coarse auctions. Because these first-score auctions implement the same order as their correspond second-score auctions, the equilibrium structure is defined by the ``break-even contract order'' with respect to $\Phi$. 

Theorem \ref{thm:rev_equiv} is not very difficult to prove. Informally, as the distribution of the opponent's type gets more competitive---as it puts more density on types with lower costs---an agent's probability of winning and his gains from shading his bid drop. Thus, his optimal bid approaches his break-even contract regardless of his type. As a result, the implemented equilibrium order approaches the break-even contract order. Then, because the equilibrium structure must be the same at all distributions, the break-even contract order must hold at all distributions. 

Applying this theorem gives the following testable condition:
\begin{proposition}\label{prop:fsa_linear}
 A first-score auction with scoring rule $\Phi$ admits a coarse beliefs equilibrium if and only if its break-even effort function, $e_{BE}:={q_{BE}}^\eta:\Theta\to [0,1]$ is linear in fixed cost.
\end{proposition}

Revenue Equivalence and Proposition \ref{prop:fsa_strat} imply that for every type, taking break-even efforts and applying expectations must commute. A converse to Jensen's inequality tells us that this implies linearity. Importantly, this test does not require computing an equilibrium. Break-even strategies are easily verifiable given any scoring rule, as shown below:

\begin{corollary}
First-score auctions with quasilinear scoring rules or price-per-quality ratio scoring rules admit CBE's. First-score auctions with a quality discount scoring rule do not.
\end{corollary}

\begin{proof}
    \begin{enumerate}[(a)]
        \item In a first-score auction with a quasilinear scoring rule, $e_{BE}$ minimizes
        \[\phi(e^{1/\eta})-p \text{ s.t. } p-e\cdot m-f\ge 0.\]
        Taking a first-order condition gives
        \[\phi'({e_{BE}}^{1/\eta})\cdot\left(\frac{1}{\eta}{e_{BE}}^{1/\eta-1}\right)-m=0.\]
        $e_{BE}$ does not depend on fixed costs so it is linear in fixed cost (with coefficient $0$).
        \item In a first-score auction with a PQR scoring rule and $\eta>1$, $e_{BE}$ minimizes
        \[\frac{p}{e^{1/\eta}} \text{ s.t. } p-em-f\ge 0.\]
        The minimand of this expression is
        \[e_{BE}(m,f)=\frac{f}{(\eta- 1)m}\]
        which is linear in $f$.
        \item In a first-score auction with a quality discount scoring rule, $e_{BE}$ minimizes
        \[p(1-e^{\frac 1 \eta})\text{ s.t. } p-em-f\ge 0.\]
        The minimand of this expression satisfies
        \[m={e_{BE}(m,f)}^{\frac 1 \eta-1}\cdot \bigg[(1+\frac 1 \eta )\cdot {e_{BE}(m,f)}\cdot m+\frac f \eta)\bigg]\]
        and is nonlinear in $f$.
        
    \end{enumerate}

\end{proof}

\section{Microfoundations}\label{sec:microfound}
In this section, I provide a formal model of information acquisition and belief formation in a coarse beliefs equilibrium. The model in this section is more general than that required to discuss concepts in Sections \ref{sec:model} and \ref{sec:results}. First, I relax assumptions on the number of and symmetry between agents. Second, I model agents playing a general game with potentially more dependence of one agent's utility on another's type and action than in a first-score auction. Third, I allow for agents to learn from a general information technology. In the main result of this section, I relate coarse prior equilibria to equilibria of common prior games and show that for moment-based information technologies, a coarse prior equilibrium is equivalent to an equilibrium in $n$-moment strategies.

\subsection{Model} \textbf{Overview.} There is a finite number of agents $n$ who play in a predetermined mechanism. At the start of the game, agent $i$'s private information is captured by the payoff type $\theta_i\in \Theta$  for $i\in\{1,...,n\}.$ There is a state of the world $g_0\in \Delta(\Theta)^n$ which captures the true distribution of payoff types in a large population from which these agents are drawn. Each agent's payoff type is, then, drawn independently according to the density $g_{0,i}\in \Delta(\Theta)$\footnote{Note that this formulation implies conditional independence. The realization of $\theta_i$ yields no information about other agents' components of the state.}. Agents know that $\mathbf g_0\in\mathcal G\subset \Delta(\Theta)^n$ but they have no priors over $\mathcal G$. They are, however, able to learn coarse information about $\mathbf g_0$. The structure of their learning process will be described in more detail below. Once the information acquisition stage is completed, the agents play in the mechanism, which is defined by the set of actions $A$, set of outcomes $Y$, and a function from $A^n\to Y$. Each agent's ex post utility from the mechanism is captured by the indirect utility function
\[U_i(\bm \theta,\mathbf a).\]
 In summary, including the information acquisition stage, they participate in a two-stage game whose timing is defined as follows:
\begin{enumerate}
\item The rules of the mechanism are announced, agents learn their payoff types, and then agents form arbitrary beliefs.
\item Each agent privately selects a feasible signal structure which is a partition of the state space $\mathcal G$.
\item Each agent receives a private signal which is the element of his selected partition that contains the true state $g_0$.
\item Each agent updates his beliefs given his signal realization.
\item The agents participate in the announced mechanism.
\end{enumerate}

 Given that the choice of initial beliefs is uninformed, my notion of simplicity will correspond to an equilibrium in which this choice is strategically irrelevant for every possible realization of signals. In other words, for every belief that is reasonable for an agent to hold, his best-response should be the same.\\
In formalizing these concepts, I will first discuss the structure of agents' learning process, then I will turn to agents' prior and posterior beliefs, and finally, I will present my solution concept.

\noindent\textbf{Signal Structures.} There is a set of (non-fully revealing) partitions of the state space $\mathcal I$. This set, called the information technology, is common knowledge. After his payoff type is realized, each agent privately selects a partition from $\mathcal I$ as his signal structure. Because there is uncertainty about other agents' payoff types, one agents' off-path behavior---that is, his choice of a signal structure at an unrealized payoff type---will affect another's agents on-path actions. Thus, for each possible payoff type, an agent chooses a signal structure. Agent $i$'s information acquisition strategy is captured by the function $I_i:\Theta\to\mathcal I$. For each payoff type $\theta_i$, the image of this map, $I_i(\theta_i)$ forms a partition of the state space. Each agent receives a signal and learns the set in $I_i(\theta_i)$ containing the true state $g_0$. Fixing a payoff type $\theta_i$, agent $i$ will have different beliefs depending on the realization of the signal. As a result, the payoff type space $\Theta$ is insufficient for fully describing the beliefs agents hold before playing in the mechanism. Thus, for each agent we define a new type space, $\mathcal T_i\subset \Theta\times 2^{\mathcal G}$ whose elements are the terms
\[(\theta_i,E_i)\]
for all $\theta_i\in \Theta$ and $E_i\in I_i(\theta_i)$. I refer to $t_i\in \mathcal T_i$ as agent $i$'s belief type. I write $\theta(t_i)$ as $t_i$'s corresponding payoff type and $E(t_i)$ as the set of distributions admissible to $t_i.$ Note that this definition of the belief type space depends on the signal structure $\bm I$. When necessary, I write the type space as $\mathcal T_i(I_i)$ but suppress this dependence when clear. 

\noindent\textbf{Beliefs.} Each agent enters the game with arbitrary beliefs about the state and their opponents' types. Rather than modelling these explicitly, I follow \cite{harsanyi} and model this implicitly through the construction of a belief structure (see \cite{mertenszamir} for more detail on the equivalence of these two approaches). Let $\pi_i\in \Delta(\mathcal G)$ be agent $i$'s (arbitrary) first-order belief\footnote{Important measurability considerations are obscured here as in general the set $\mathcal G$ can be infinite-dimensional. To deal with this, we first restrict to distributions whose supports admit some finite parametrization.  For simplicity, we then restrict the set of possible priors to measures that are absolutely continuous with respect to the Lebesgue measure over this parameter space.}. After choosing the signal structure $I_i$ in the first stage and learning the realization $E_i$, each agent updates his first-order beliefs using Bayes rule. In other words, this signal realization induces the posterior first-order beliefs  $\tilde \pi_i:\mathcal T_i\to \Delta(\mathcal G)$ defined by
\[\mathbb P_{\tilde \pi_i(\theta_i,E_i)}(A)=\frac{\mathbb P_{\pi_i}(A\cap E_i)}{\mathbb P_{\pi_i}(E_i)}\]

for all Lebesgue measurable sets $A\subset \mathcal G$. In general, defining first-order beliefs is insufficient for fully specifying belief hierarchies. However, because the state in this model is a distribution over opponents' types, each first-order belief structure uniquely extends to a full hierarchy of beliefs as expressed below:

\begin{lemma}\label{lem:belief_extend}
The first-order belief structure $\bm{\tilde \pi}$ uniquely defines beliefs of the form $\hat \pi_i:\mathcal T_i(I_i)\to \Delta(G\times \mathcal T_{-i}( I_{-i}))$ for all $i$.
\end{lemma} 

\noindent\textbf{Solution Concept.} Once agents have formed posterior beliefs, they are expected utility maximizers. Thus, fixing some posterior belief structure $\tilde \pi$ and opponents' action profile $a_{-i}:\mathcal T_{-i}\to A^{n-1}$, an agent with type $t_i$ prefers action $a_1$ over $a_2$ in the second stage of the game if  
\[\E_{\tilde\pi}[U(\theta(t_i),\theta(t_{-i}),a_1,a_{-i}(t_{-i})]\ge \E_{\tilde\pi}[U(\theta(t_i),\theta(t_{-i}),a_2,a_{-i}(t_{-i})].\]
To define preferences over the whole game, intuitively an agent should prefer one signal structure to another if it allows him to choose an action that gives him a higher expected utility: 

\begin{definition}[Preferences] Fix prior beliefs $\bm \pi$ and fix the strategy of other agents to be $\mathbf I_{-i}:\Theta^{n-1}\to \mathcal I^{n-1}$ and $a_{-i}: \mathcal T_{-i}(\mathbf I_{-i})\to A^{n-1}$. Then agent $i$ prefers the strategy $(I_i^1,a^1)$ to $(I_i^2,a^2)$ if for all  payoff types $\theta_i\in \Theta$,
\[\E_{\pi}[U(\theta_i,\theta(t_{-i}(g)),a^1(t^1_i(g)),a_{-i}(t_{-i}(g))]\ge \E_{\pi}[U(\theta_i,\theta(t_{-i}(g)),a^2(t^2_i(g)),a_{-i}(t_{-i}(g))]\]
where $t_{-i}(g)$ represents opponents' belief types at $g$ and $t^j_i(g)$ represents agent $i$'s belief type at $g$ under the signal structure $I_i^j(\theta_i)$. If this holds, we write that
\[(I_i^1,a^1)\succeq_{i,(\mathbf I_{-i},a_{-i}),\pi} (I_i^2,a^2)\]
\end{definition}
Because each agent arbitrarily chooses beliefs at the beginning of the game, his preferences should not be dependent on this random selection. Thus, we can refine this definition of preferences to one that holds at all beliefs:

\begin{definition}[Strong Preferences] 
Agent $i$ strongly prefers the strategy $(I_i^1,a^1)$ to $(I_i^2,a^2)$ if for all possible belief structures $ \tilde{\bm \pi}$, we have that
\[(I_i^1,a^1)\succeq_{i,(\mathbf I_{-i},a_{-i}),\pi} (I_i^2,a^2).\]
In this case, we write 
\[(I_i^1,a^1)\succeq_{i,(\mathbf I_{-i},a_{-i})} (I_i^2,a^2).\]
\end{definition}

I use the term preference but this relation is not a preference relation in the standard sense. In particular, it is not rational because it is not complete. One strategy may be better than another at some but not all prior belief structures. As a result, there may not be a unique maximal element according to this ordering. In general, the existence of best-responses depends both on the definition of the information technology and the mechanism. Nevertheless, when best responses do exist, a corresponding equilibrium notion can be defined.

\begin{definition}
A mechanism admits a \textbf{coarse beliefs equilibrium} with respect to the information technology $\mathcal I$ if for each agent, there is an information acquisition strategy $I_i^*:\Theta\to\mathcal I$ and a second-stage strategy $a_i^*:\mathcal T(I_i^*)\to A$ such that
\[(I_i^*,a_i^*)\succeq_{i,(I_{-i}^*,a_{-i}^*)} (I,a)\]
for all other strategies $I:\Theta\to\mathcal I$ and $a:\mathcal T(I)\to A$.
\end{definition}

The simplest example of a coarse beliefs equilibrium is when the information technology is fully uninformative. In this case, an agent's equilibrium strategy must be a best response at any belief given equilibrium strategies of his opponents. In this case, this concept is the same as a robust equilibrium. In general, the concept of a coarse beliefs equilibrium can be viewed as a weakening of the concept of a robust equilibrium. The space of feasible beliefs is restricted through the receipt of signals but after that, it requires an action be optimal at all reasonable belief structures.

\textit{Remark.} The common knowledge assumptions in this model are substantial weakenings of those in the Bayes-Nash model. In particular, this model does not require agents share a common prior or even common knowledge of some information structure. All that is needed is common knowledge about the information technology $\mathcal I$. I view this as less restrictive than specific knowledge about belief hierarchies.

\subsection{Results}

As mentioned above, the concept of a CBE relates closely to robust equilibrium. As shown in \citet{bergmorr}, in certain classes of games, a mechanism satisfies robust implementability if and only if it is interim-implementable on all common prior type spaces. An analogous result will hold here. 

To begin, I define the common prior game at $g\in\mathcal G$ to be similar to the one defined above with the sole change that it is common knowledge that all agents believe the true distribution is $g$. All agents are assumed to be expected utility maximizers so they play pure Bayes-Nash strategies written as  $a^{CP}_i(\theta_i,g)$\footnote{I assume a unique, pure strategy Bayes-Nash equilibrium in this section solely for ease of exposition. These results hold more generally.}. The main result of this section can be stated as follows:

    \begin{theorem}\label{thm:CBE_informal}
        A mechanism admits a coarse beliefs equilibrium defined by $I_i^*$ and $a_i^*:\mathcal T(I_i^*)\to A$ if and only if
        \[a_i^*(t_i)=a^{CP}_i(\theta_i,g_1)=a^{CP}_i(\theta_i,g_2)\]
        for all $t_i:=(\theta_i,E_i)\in\mathcal T( I_i^*)$ and $g_1,g_2\in E_i$.\end{theorem}

The proof of Theorem \ref{thm:CBE_informal} is quite straightforward. For the forward direction, if a strategy is optimal at all beliefs, then it must be optimal at all plausible common prior beliefs. For the reverse direction, if an action is optimal point-wise at every type distribution $g$, then it must be optimal at any belief that can be expressed as a mixture over these distributions. 

Thus, instead of considering this two-stage game and many different belief structures, it suffices to characterize equilibrium actions in common prior games to understand which equilibria can be implemented in coarse beliefss. Implementability of an equilibrium depends on the information technology to which agents have access. I will consider from now on the symmetric case in which all agents' types are drawn from the same distribution. Slightly, abusing notation I will call this distribution $g_0$ and let $\mathcal G=\Delta(\Theta)$. I define a natural class of information technologies below:

\begin{definition}\label{def:n_mom_info_tech}
The $n$-moment information technology $\mathcal I_n$ is defined as
\[\mathcal I_n=\{(M_1,...,M_n):n\in\mathbb N,M_i\in \mathcal G^*\}.\]
The finite-moment information technology is defined as
\[\mathcal I_{\mathbb N}=\bigcup_{n\in\mathbb N} \mathcal I_n.\]
\end{definition}
The $n$-moment information technologies are not formally defined through partitions of the state space. The objects in definition \ref{def:n_mom_info_tech} can be identified with information technologies through the following relationship:
\[(M_1,...,M_n)\Longleftrightarrow \bigg\{\{g\in\mathcal G:M_i(g)=r_i~\forall i\}:\forall (r_1,...,r_n)\in\mathbb R^N\bigg\}\]
In words, a signal structure in $\mathcal I_{\mathbb N}$ specifies a finite sequence of moments. We identify a sequence of moments with the partition of the state space induced through different moment realizations. With such a definition, we can formally discuss concepts that appear in previous sections of this paper.

\begin{corollary}\label{cor:fmit_nmoments}
A first-score auction has a coarse beliefs equilibrium with respect to $\mathcal I_{n}$ if and only if its common prior equilibria can be implemented in $n$-moment strategies.
\end{corollary}

As discussed in Section \ref{sec:model}, the main results of this paper imply auctions that have a CBE with respect to $\mathcal I_{\mathbb N}$ also have CBE's with respect to strictly less informative information technologies:

\begin{corollary}
Let $\Phi^*(\mathcal I)$ be the set of scoring rules that have coarse beliefs equilibria with respect to $\mathcal I$. Then
\[\Phi^*(\mathcal I_k)=\begin{cases}
\emptyset&
    k\in\{0,1\}\\
    \Phi^*(\mathcal I_{\mathbb N})&k>2.
\end{cases}\]
\end{corollary}

\color{black}

\section{Discussion}
\textbf{Number of participants.} A primary limitation of the model as it is written is that only $2$-participant auctions are considered formally. However, this is without much loss; for any symmetric first-score auction with $N+1$ participants, there is a strategically-equivalent $2$-player auction. To see this, consider the simpler first-price auction. It makes no difference strategically whether an agent is competing against $N$ bidders drawn from a density $g$ or $1$ strong bidder whose value is drawn from a density $Ng^{N-1}$. Returning to the first-score context, if the distribution of the first-order statistic of the opponents' equivalence classes are equal across two settings then strategies will be the same. The two moments used in CBE strategies would be taken with respect to this distribution in a general $N+1$ participant setting\footnote{This type of moment could be interpreted with the addition of a resampling phase to a standard estimation procedure. A sample of relevant firms would first be collected. Then a researcher could sample sets of $N-1$ firms from this data and keep track of the highest-scoring type of each set (according to the break-even contract order). Estimating means with respect to the distribution of these types would yield the desired moments.}. Extending the definition of CBE to this setting, auction formats would be evaluated in the same way as they are in this simpler model.

\textbf{Other benefits of equilibrium invariance.} The characterization of auctions that admit CBE's shows that they are exactly the ones that retain some nice properties of one-dimensional first-price auctions. Thus, it seems possible that this specific notion of simplicity would map onto other appealing properties of a mechanism. Two examples are discussed here. 

Firstly, one could consider an agent who interprets data from past auctions and plays in a current iteration of the same format\footnote{Many auctions in practice release information on winning bids to help participants strategize. See \citet{andymartino} for a discussion of this practice in Google ad auctions.}. If the past auctions have variable numbers of participant, the agent needs a model of interpreting data from those auctions in the setting he faces. This would be difficult in general, because as the number of participants changes, the equilibrium structure changes. Equilibrium invariance may give the agent a natural way to avoid this problem and to be able to interpret previous auction data, regardless of the number of participants.

Secondly, in general first-score auctions, \citet{n-h} show that a fixed-point problem must be solved to determine equilibrium strategies. Because of the well-established difficulty of doing so, a computationally-constrained agent would likely be unable to strategize optimally. If this holds, the only equilibrium strategies he could find would be in auctions that satisfy equilibrium invariance\footnote{See \citet{fphard} for a discussion on the difficulty of solving fixed-point problems.}.

\bibliographystyle{plainnat}
\bibliography{main}

\begin{thebibliography}{28}
\providecommand{\natexlab}[1]{#1}
\providecommand{\url}[1]{\texttt{#1}}
\expandafter\ifx\csname urlstyle\endcsname\relax
  \providecommand{\doi}[1]{doi: #1}\else
  \providecommand{\doi}{doi: \begingroup \urlstyle{rm}\Url}\fi

\bibitem[Asker and Cantillon(2008)]{ac08}
John Asker and Estelle Cantillon.
\newblock Properties of scoring auctions.
\newblock \emph{Rand Journal of Economics}, 39\penalty0 (1):\penalty0 69--85, 2008.

\bibitem[Banchio and Skrzypacz(2022)]{andymartino}
Martino Banchio and Andrzej Skrzypacz.
\newblock Artificial intelligence and auction design.
\newblock In \emph{EC '22: Proceedings of the 23rd ACM Conference on Economics and Computation}, pages 30--31. Association for Computing Machinery, 2022.

\bibitem[Battigalli and Siniscalchi(2003)]{rationalizability}
Pierpaolo Battigalli and Marciano Siniscalchi.
\newblock Rationalization and incomplete information.
\newblock \emph{Advances in Theoretical Economics}, 3\penalty0 (1):\penalty0 1073--1116, 2003.

\bibitem[Bergemann and Morris(2005)]{bergmorr}
Dirk Bergemann and Steven Morris.
\newblock Robust mechanism design.
\newblock \emph{Econometrica}, 73:\penalty0 1521--34, 2005.

\bibitem[Bergman and Lundberg(2013)]{lundberg}
Mats~A. Bergman and Sofia Lundberg.
\newblock Tender evaluation and supplier selection methods in public procurement.
\newblock \emph{Journal of Purchasing and Supply Management}, 19\penalty0 (2):\penalty0 73--83, 2013.

\bibitem[B\"orgers and Li(2019)]{borgers-li}
Tilman B\"orgers and Jiangtao Li.
\newblock Strategically simple mechanisms.
\newblock \emph{Econometrica}, 87\penalty0 (6):\penalty0 2003--2035, 2019.

\bibitem[Brooks and Du(2023)]{brooksdu}
Benjamin Brooks and Songzi Du.
\newblock Maxmin auction design with known expected values.
\newblock Working paper, 2023.

\bibitem[Camara(2022)]{modibo}
Modibo Camara.
\newblock Mechanism design with a common dataset.
\newblock In \emph{EC '22: Proceedings of the 23rd ACM Conference on Economics and Computation}, page 558. Association for Computing Machinery, 2022.

\bibitem[Che(1993)]{che}
Yeon-Koo Che.
\newblock Design competition through multidimensional auctions.
\newblock \emph{Rand Journal of Economics}, 24\penalty0 (4):\penalty0 668 -- 680, 1993.

\bibitem[Dastidar(2014)]{dastidar}
Krishnendu~G. Dastidar.
\newblock Scoring auctions with non-quasilinear scoring rules.
\newblock Working paper, 2014.

\bibitem[Etessami(2014)]{fphard}
Kousha Etessami.
\newblock A taxonomy of fixed point computation problems for algebraically-defined functions, and their computational complexity.
\newblock In \emph{Simons Institute Workshop: Solving Polynomial Equations}. Simons Institute for the Theory of Computing, 2014.

\bibitem[{Federal Acquisition Regulatory Council}(1983)]{fed-code}
{Federal Acquisition Regulatory Council}.
\newblock Sealed bidding.
\newblock Code of Federal Regulations, 1983.

\bibitem[Gretschko and Mass(2024)]{gretschko-mass}
Vitali Gretschko and Helene Mass.
\newblock Worst-case equilibria in first-price auctions.
\newblock \emph{Theoretical Economics}, 19\penalty0 (1):\penalty0 61--93, 2024.

\bibitem[Hanazono et~al.(2016)Hanazono, Hirose, Nakabayashi, and Tsuruoka]{n-h}
Makoto Hanazono, Yohsuke Hirose, Jun Nakabayashi, and Masanori Tsuruoka.
\newblock Theory, identification, and estimation for scoring auctions.
\newblock Working paper, 2016.

\bibitem[Harsanyi(1968)]{harsanyi}
John Harsanyi.
\newblock Games with incomplete information played by bayesian play ers, i-iii.
\newblock \emph{Management Science}, 14:\penalty0 159--182, 320--334, 486--502, 1968.

\bibitem[Hayek(1945)]{hayek}
Friedrich.~A. Hayek.
\newblock The use of knowledge in society.
\newblock \emph{The American Economic Review}, 35\penalty0 (4):\penalty0 519--530, 1945.

\bibitem[Li and Dworczak(2024)]{dworczak-li}
Jiangtao Li and Piotr Dworczak.
\newblock Are simple mechanisms optimal when agents are unsophisticated?
\newblock Working paper, 2024.

\bibitem[Li(2017)]{osp}
Shengwu Li.
\newblock Obviously strategy-proof mechanisms.
\newblock \emph{American Economic Review}, 107\penalty0 (11):\penalty0 3257--87, 2017.

\bibitem[Liang(2021)]{annie}
Annie Liang.
\newblock Games of incomplete information played by statisticians.
\newblock Working paper, 2021.

\bibitem[Matejka and McKay(2015)]{matejkamckay}
Filip Matejka and Alisdair McKay.
\newblock Rational inattention to discrete choices: A new foundation for the multinomial logit model.
\newblock \emph{American Economic Review}, 105\penalty0 (1):\penalty0 272--298, 2015.

\bibitem[Mertens and Zamir(1985)]{mertenszamir}
J-F. Mertens and S.~Zamir.
\newblock Formulation of bayesian analysis for games with incomplete information.
\newblock \emph{International Journal of Game Theory}, 14:\penalty0 1--29, 1985.

\bibitem[Morris(1995)]{morris95}
Stephen Morris.
\newblock The common prior assumption in economic theory.
\newblock \emph{Economics and Philosophy}, 7\penalty0 (2):\penalty0 227--253, 1995.

\bibitem[Oll\'ar and Penta(2017)]{ollarpenta}
Mariann Oll\'ar and Antonio Penta.
\newblock Full implementation and belief restrictions.
\newblock \emph{American Economic Review}, 107\penalty0 (8):\penalty0 2243--2277, 2017.

\bibitem[Pearce(1984)]{pearce}
David Pearce.
\newblock Rationalizable strategic behavior and the problem of perfection.
\newblock \emph{Econometrica}, 52:\penalty0 1029--1050, 1984.

\bibitem[Pernoud and Gleyze(2023)]{agatheisi}
Agathe Pernoud and Simon Gleyze.
\newblock Informationally simple incentives.
\newblock \emph{Journal of Political Economy}, 131\penalty0 (3):\penalty0 802--837, 2023.

\bibitem[Verens(2021)]{eu}
Rudolf Verens.
\newblock The state of implementing procurement procedures in eu agencies: enhancing transparency and assessing flexibility.
\newblock Directorate-General for Internal Policies, 2021.

\bibitem[Wang and Liu(2014)]{wang-liu}
Mingxi Wang and Shulin Liu.
\newblock Equilibrium bids in practical multi-attribute auctions.
\newblock \emph{Economics Letters}, 123\penalty0 (3):\penalty0 352--355, 2014.

\bibitem[Wolitzky(2016)]{wolitzky}
Alexander Wolitzky.
\newblock Mechanism design with maxmin agents: Theory and an application to bilateral trade.
\newblock \emph{Theoretical Economics}, 11\penalty0 (3):\penalty0 971--1004, 2016.

\end{thebibliography}

\pagebreak

\appendix
\section{Proofs from Section \ref{sub:fsa_strat}}
Throughout this appendix, it will be useful to think of the extended type space $\overline \Theta:= [\underline m,\overline m]\times \mathbb R$ and extend all relevant definitions to this space. We do this because given some type $(m_1,f_1)\in \Theta$, it may be that the type in the same equilibrium equivalence class with marginal cost $m_2\in [\underline m,\overline m]$ does not have fixed cost in $[\underline f,\overline f]$. We'll consider strategies of types in $\overline \Theta$ even though the only types with positive density are those in $\Theta$.

In this section I prove Proposition \ref{prop:fsa_strat} and its corollary. I start with a lemma that will be useful in the proof of Proposition \ref{prop:fsa_strat}.

\begin{lemma}\label{lem:marg_score}
    Fix a type $\theta=(m,f)\in \Theta$. There is some type $(m,f_{(2)})\in \overline \Theta$ such that $s^*(m,f)=s_{BE}(m,f_{(2)})$ where $s_{BE}=\Phi\circ (p_{BE},q_{BE}).$
\end{lemma}
\begin{proof}
    Firstly, we have that the score bid by this type in equilibrium satisfies \[s^*(m,f)\in [s_{BE}(m,\overline f+\overline m),s_{BE}(m,f)].\]
    
    To see this, we have that $s^*(m,f)\le s_{BE}(m,f)$ because bidding above this break-even score guarantees a strictly negative payoff conditional on winning the auction and a strictly positive probability of winning. A profitable deviation is to simply bid the break-even score. \\ In equilibrium, the lowest-type agent in $\Theta$ will bid his break-even score and win with probability $0$. Any other type will bid above this score to obtain a strictly positive expected payoff. Thus, it suffices to show that 
    \[s_{BE}(m,\overline f+\overline m)\le s_{BE}(\overline m,\overline f).\]
    If $\overline p,\overline q$ is the break-even contract of the type $(m,\overline f+\overline m)$, then because $\bar q\le 1$,
    \[0\le \overline p - m\cdot \overline q^\eta -(\overline f +\overline m) < \overline p -\overline m \cdot \overline q^\eta -\overline f \]
    The type $(\overline m,\overline f)$ can obtain the score $s_{BE}(m,\overline f+\overline m)$ with a contract that provides positive payoff so 
    $s_{BE}(m,\overline f+\overline m)\le s_{BE}(\overline m,\overline f)$.

    Now, we have that the break-even score for a type $(m,f)$ is the solution to the maximization problem\footnote{It is without loss to restrict the constraint to the equality case because the scoring rule is strictly decreasing in price.}
    \[\max_{p,q\in\mathcal C} \Phi(p,q) \text{ s.t. } p= m\cdot q^{\eta}+f.\]
    The constraint correspondence is compact-valued and continuous in $(m,f)$ so by Berge's Theorem of the Maximum, the solutions to this maximization problem are continuous in $(m,f)$. Then, if we fix a marginal cost $m$ and apply the Intermediate Value Theorem to the continuous function $s_{BE}(m,\cdot):[\underline f,\overline f +\overline m]\to \mathbb R$, there must be some $f_{(2)}$ satisfying the desired condition.
\end{proof}

Now, we turn attention to Proposition \ref{prop:fsa_strat}.

\textbf{Proposition \ref{prop:fsa_strat}.}\textit{
Fix $g\in \mathcal G$. For all $\theta\in \Theta$, the equilibrium strategy can be represented as 
                \[(p^*,q^*)(m,f,g)=(p_{BE},q_{BE})(m,f_{(2)})\]
where
                \begin{equation*}f_{(2)}=\frac{\int z\cdot \mathbb I_{z\ge f}~ g_m([(m,z)]_g)~ dz}{\int \mathbb I_{z\ge f} ~g_m([(m,z)]_g) dz}.\end{equation*}
}

\begin{proof}
As discussed in the main text, we would like to project our auction with $2$ dimensions of private information into a more familiar one-dimensional analogue where all types have marginal cost $m$. From Proposition $2$ of \cite{n-h}, we have that at any $g\in\mathcal G$, there exists a continuous density function of scores bid in equilibrium, $h(s)$ (because we are dealing with a fixed $g$ currently, we suppress dependence of $h$ on $g$ in the notation).  Now, fixing an $m$, we define a new distribution over types $g_{m}:\mathbb R\to\mathbb R$ as follows:
  \[G_m(f):=H(s^*(m,f,g))\]
  \[g_{m}(f):=h(s^*(m,f,g))\cdot \frac{\partial s^*(m,f',g)}{\partial f'}\bigg\rvert_{f'=f}.\] 
  
We now consider an auction in which types' fixed costs are their only private information and are drawn i.i.d according to $G_m$. In this auction, the agent with cost $f$ who receives their bid $b$ upon winning attains the expected utility
\[\mathbb P(win|b)\cdot (b-f).\]
From standard results on symmetric first-price auctions, in the unique symmetric Bayes-Nash equilibrium of this auction, bids satisfy
\begin{equation}\label{eq:1dauc} G_m(f)\cdot {b^*}'(f) +g_m(f)\cdot(b^*(f)-f)=0.\end{equation}
Now consider the original scoring auctions. If the type $(m,f)$ bids the score $s,$ then best-responding implies that  
\begin{equation}\label{eq:2dauc}H(s)\cdot u_1(s,m,f)+h(s)\cdot u(s,m,f)=0\end{equation}
where $u(s,\theta)$ is the indirect utility to type $\theta$ from bidding $s$. 
If we let $b(f)=p^*(m,f,g)-m\cdot {q^*}^\eta(m,f,g)$, then 
\[u(s,m,f)=b(f)-f \text{ and } u_1(s,m,f)\cdot s_2^*(m,f,g)=b'(f).\]
Simplifying gives us that Equations \ref{eq:1dauc} and \ref{eq:2dauc} are equivalent. Thus, 
\[b(f)=b^*(f):=\frac{\int z\cdot \mathbb I_{z\ge f}~ g_m(f) dz}{\int \mathbb I_{z\ge f} ~g_m(f) dz}.\]
 From Lemma \ref{lem:marg_score}, the contract $(p^*(m,f,g),q^*(m,f,g))$ is the break-even strategy of some type $(m,f_{(2)})$. We finish by noting that by the definition of a break-even strategy,
\[b^*(f)=p^*(m,f,g)-m\cdot {q^*}^\eta(m,f,g)=p_{BE}(m,f_{(2)})-m\cdot {q^\eta}_{BE}(m,f_{(2)})=f_{(2)}.\]
\end{proof}

\begin{corollary}
    If an auction satisfies equilibrium invariance then it has an informational precision of $2$.
\end{corollary}
\begin{proof}
Fixing a marginal cost $m\in [\underline m,\overline m]$, we can define a projection function that maps the equivalence classes of types into the type with marginal cost $m$:
\[\rho_m:\Theta\to [\underline f,\overline f +\overline m]\]
\[\theta \mapsto f \text{ such that } (m,f) \in [\theta]_g\]
This is a well-defined function because no two types with the same marginal cost will be in the same equivalence class. We can write the expression in Equation \ref{eq:f2} as 
\[f_{(2)}=\frac{\int \rho_m(\tau)\cdot \mathbb I_{(m,f)\succeq \tau}~ g(\tau)d\tau}{\int \mathbb I_{(m,f)\succeq \tau}~ g(\tau)d\tau}\]
which is a smooth function of $2$ moments\footnote{Note that we suppress dependence of $\succeq$ on $g$ when equilibrium invariance is satisfied.}. The differentiability of the functions $p_{BE}$ and $q_{BE}$ results from the smoothness of the scoring rule. 

To see that no such auction can have precision less than $2$, no first-score auction can have precision $0$ by Proposition \ref{prop:fsa_strat} because strategies depend on the distribution $g$. To show an auction satisfying equilibrium invariance cannot have an informational precision of $1$, fix a type $\theta$ and assume it has a strategy with a precision of $1$, described by the function $\zeta_\theta:\Theta\to \mathbb R$. Fix a distribution $g_1\in\mathcal G$ and let $v\in\mathcal G_0$ satisfy
\[\int \zeta_\theta(\tau) v(\tau)~ d\tau = 0.\]
In words, the strategy of type $\theta$ should not change when any distribution is changed in the direction $v$. Then it must be that
\[\int (\rho_m(\tau)-f_{(2)})\mathbb I_{(m,f) \succeq \tau} v(\tau)~ d\tau = 0\]  
where $f_{(2)}=\frac{\int \rho_m(\tau)\cdot \mathbb I_{(m,f)\succeq \tau}~ g_1(\tau)d\tau}{\int \mathbb I_{(m,f)\succeq \tau}~ g_1(\tau)d\tau}$. Otherwise, from Proposition \ref{prop:fsa_strat}, this type's strategy would change when the distribution $g_1$ changes in the $v$ direction. 
Now pick some distribution $g_2$ so that 
\[f_{(2)}\not=\frac{\int \rho_m(\tau)\cdot \mathbb I_{(m,f)\succeq \tau}~ g_2(\tau)d\tau}{\int \mathbb I_{(m,f)\succeq \tau}~ g_2(\tau)d\tau}.\]
Then changing $g_2$ in the direction $v$ should change type $\theta$'s strategy by Proposition \ref{prop:fsa_strat} but the strategy with precision $1$ does not change, giving a contradiction. 
\end{proof}

\section{Proofs From Section \ref{sub:diff}}
The goal of this section is to prove Proposition \ref{prop:converse}. There are $2$ lemmas proven beforehand to ease in exposition. The first tells us that if a type bids the same score at two distributions $g_1$ and $g_2\in \mathcal G$, then so will all types in the corresponding equivalence class. The second is the lemma discussed in the main text, essentially implying that the moments used in a finite-precision strategy must correspond to the equilibrium structure. 

    \begin{lemma}\label{lem:curve_fixed}
        Let $\theta\in \Theta$ and $g_1,g_2\in\mathcal G$. If $s^*(\theta,g_1)=s^*(\theta,g_2)$, then
        \[s^*(\theta',g_1)=s^*(\theta',g_2)=s^*(\theta,g_1)\]
        for all $\theta'\in [\theta]_{g_1}$. Equivalently,
        \[[\theta]_{g_1}=[\theta]_{g_2}.\]
    \end{lemma}
    \begin{proof}
         Let $h_1,h_2$ be the equilibrium score distributions at $g_1$ and $g_2$ respectively. Optimizing expected utility implies that
        \[H_1(s)\cdot u(s,\theta')+h_1(s)\cdot u_1(s,\theta')=0\]
        for all $\theta'\in [\theta]_{g_1}$ where $s=s^*(\theta,g_1)$. From the given,
        \[H_1(s)\cdot u(s,\theta)+h_1(s)\cdot u_1(s,\theta)=H_2(s)\cdot u(s,\theta')+h_1(s)\cdot u_1(s,\theta')=0.\]
        Thus, $\frac{h_1(s)}{H_1(s)}=\frac{h_2(s)}{H_2(s)}$ and 
        \[H_2(s)\cdot u(s,\theta')+h_2(s)\cdot u_1(s,\theta')=0\]
        for all $\theta'\in [\theta]_{g_1}$. There is a unique type with a given marginal cost for which this FOC holds. Thus, it is sufficient for concluding that $\theta'\in [\theta]_{g_2}$.
    \end{proof}

This will be applied in the proof of the following lemma from the main text:

\noindent \textbf{Lemma \ref{lem:moment_match}.}\textit{
Assume the auction can be implemented with $n$-moment strategies and fix a type $\theta\in\Theta$ and a distribution $g\in\mathcal G$. Then, the following must hold:
\begin{enumerate}[(a)]
\item Except for on a set of measure $0$, $\psi_{\theta,g}(\tau)$ is constant on $\{\tau\in \Theta:\tau\succ_g \theta\}$.
\item Let $c$ be the value of $\psi_{\theta,g}(\tau)$ on this set. Then there exists a type $\theta'\prec_g \theta$ such that, almost surely, \[\psi_{\theta,g}(\tau)\not=c\text{ for }\{\tau\in\Theta:\theta' \prec_g \tau\prec_g \theta\}.\]
\end{enumerate}}
\begin{proof}
\begin{enumerate}[(a)]
\item  Asssume for the sake of contradiction that $\psi_{\theta,g}(\tau)$ is nonconstant on the set \[S=\{\tau\in \Theta:\tau\succ_g \theta\}.\]
Let $\mu_\psi$ be the mean of $\psi_{\theta,g}(\tau)$ on $S$ with respect to the Lebesgue measure. Then by the Fundamental Lemma of the Calculus of Variations, there exists some $v\in\mathcal L^1(\Theta)$ such that
\[\int (\psi_{\theta,g}(\tau)-\mu_\psi) v(\tau) d\tau \not = 0 \text{ and } v(\tau)=0 \text{ for } \tau \not\in S.\]
Let let $\mu_{v}$ be the mean of $v$ on $S$. Then defining the vector $v_0\in\mathcal G_0$ with
\[v_0(\tau)=\begin{cases} v(\tau)-\mu_{v} & \tau\in S\\ 0 & \text{else}\end{cases}\]
we have that
\[\int \psi_{\theta,g}(\tau) v_0(\tau) d\tau = \int \psi_{\theta,g}(\tau) v(\tau) d\tau -\mu_\psi \cdot\mu_v=\int (\psi_{\theta,g}(\tau)-\mu_\psi) v(\tau) d\tau \not = 0.\]

Thus, by the definition of $\psi_{\theta,g}$, perturbing the distribution $g$ in the $v_0$ direction changes type $\theta$'s score bid. However, by Proposition \ref{prop:fsa_strat}, the strategy of type $\theta$ should only depend on the strategies of and the distribution over types in $S^c$ (which is unchanged by Lemma \ref{lem:curve_fixed}). Because none of these are changing due to this perturbation, type $\theta$'s strategy should not change.\footnote{To be more formal, the score bid distribution is the solution to a differential equation whose initial condition relates to the lowest score. Thus, if at some score, the distribution of lower scores is unchanged then this differential equation will be unchanged. See \cite{n-h} for more details.} 

\item For each $m\in [\underline m,\overline m]$, define the type \[f_{(2)}(m):=\frac{\int \rho_m(\tau)\mathbb I_{\tau\prec_g \theta} g(\tau)d\tau}{\int \mathbb I_{\tau\prec_g \theta} g(\tau)~d\tau}\] as in Proposition \ref{prop:fsa_strat}. Let $\theta'$ be some type satisfying
\[(m,f_{(2)}(m))\prec_g \theta' \prec_g \theta\]
for all $m$. With this type defined,
assume for the sake of contradiction that $\psi_{\theta,g}(\tau)=c$ on some set \[S\subset \{\tau\in\Theta:\theta' \prec_g \tau \prec_g \theta\}\] with positive measure. Then let $v\in\mathcal G_0$ with 
\[v(\tau)=\begin{cases} -c_1 & \tau\in \{\tau'\in \Theta:\tau'\succ_g \theta\}\\ c_2 & \tau \in S \\ 0 & \text{else}\end{cases}\]
for suitable positive constants $c_1,c_2\in\mathbb R_{++}$ such that
\[\int \psi_{\theta,g}(\tau) v(\tau) d\tau=0.\]
Perturbing a distribution in the direction $v$ moves mass from types that bid higher than $\theta$ to types that bid lower than $\theta$ but not significantly lower. The equation above implies that the strategy of type $\theta$ should not change locally if $g$ is perturbed in the direction $v$. We aim to show this is not the case. If type $\theta$'s score bid is unaffected by this perturbation then so are the bids of all types in the equivalence class $[\theta]_g$ by Lemma \ref{lem:curve_fixed}. If all of these strategies are unchanging under $v_0$, then by Proposition \ref{prop:fsa_strat}
\[\frac{\partial \left(\frac{\int \rho_m(\tau,g+\ep v)\mathbb I_{\tau\prec_{g+\ep v} \theta} (g+\ep v)(\tau)~d\tau}{\int\mathbb I_{\tau\prec_{g+\ep v} \theta} (g+\ep v)(\tau)~d\tau}\right)}{\partial \ep}\bigg\rvert_{\ep=0}=0\] for all $m$. Simplifying and noting that $\frac{\partial \int \mathbb I_{\tau\prec_{g+\ep v} \theta}(\tau)-\mathbb I_{\tau\prec_{g+\ep v} \theta}(\tau)d\tau}{\partial \ep}\bigg\rvert_{\ep=0}=0$ yields
\[\underbrace{\int \frac{\partial \rho_m(\tau,g+\ep v)}{\partial \ep} g(\tau)\mathbb I_{\tau\prec_g \theta}~d\tau}_{(\dagger)} + \int \rho_m(\theta) \mathbb I_{\tau\prec_g \theta}v(\tau)~d\tau - f_{(2)}(m)\int \mathbb I_{\tau\prec_g \theta}v(\tau)~d\tau=0. \]
By construction, $\mathbb I_{\tau\prec_g \theta}v(\tau)\not=0$ only when its value is positive and $\rho_m(\tau)>f_{(2)}(m)$. Thus, this equation implies that the first integral (denoted by $(\dagger)$) must be strictly negative for all $m$. This means that, on average, each type's equivalent that has marginal cost $m$ has a lower fixed cost upon perturbing in the $v$ direction. Intuitively, this cannot happen at all $m$. For example, let $(m,f)\sim_g (m',f')$ but let $(m,f)\prec_{g+\ep v} (m',f')$. Then
\[\rho_{m'}(m,f)<f'.\]
But necessarily, we also have that
\[\rho_{m}(m',f')>f.\]
This fact that one type becoming less competitive inherently means another type becomes more competitive is what will lead to a contradiction. Examining this more formally, we have that by the definition of $\rho_m$,

\[s^*(m,\rho_{m}(\tau,g+\ep v),g+\ep v)=s^*(\tau,g+\ep v).\]
If at $\ep=0$, we have that $\rho_m(\tau,g)=f_1$, then implicitly differentiating yields
    \[\frac{\partial s^*(m,f_1,g+\ep v)}{\partial\ep}\bigg\rvert_{\ep=0}+\left(\frac{\partial s^*(m,f,g+\ep v)}{\partial f}\bigg\rvert_{f=f_1}\right) \left(\frac{\partial\rho_{m}(\tau,g+\ep v)}{\partial\ep}\bigg\rvert_{\ep=0}\right)=\frac{\partial s^*(\tau,g+\ep v)}{\partial\ep}\bigg\rvert_{\ep=0}.\]
    Rearranging gives us that
    \[\frac{\partial\rho_{m}(\tau,g+\ep v)}{\partial\ep}\bigg\rvert_{\ep=0}=
    \frac{\frac{\partial s^*(\tau,g+\ep v)}{\partial\ep}\bigg\rvert_{\ep=0}-\frac{\partial s^*(m,f_1,g+\ep v)}{\partial\ep}\bigg\rvert_{\ep=0}}{\frac{\partial s^*(m,f,g))}{\partial f}\bigg\rvert_{f=f_1}}.\]
The terms on the RHS exist because of the differentiability assumptions on $n$-moment strategies. Returning to the expression $(\dagger)$, recall that, fixing a marginal cost $m$, we can consider the space $\Theta$ as the set of equivalence classes $[(m,f)]_g$. We can thus rewrite this integral as
\[\int \frac{\partial \rho_m(\tau,g+\ep v)}{\partial \ep} g(\tau)\mathbb I_{\tau\prec_g \theta}~d\tau=\int \int \left(\frac{\partial \rho_m(m',\rho_{m'}(m,f),g+\ep v)}{\partial \ep}~ dm'\right) \mathbb I_{(m,f)\prec_g \theta}~ g_m(f)~ df\]
where the density $g_m(f)$ satisfies
\[g_m(f)=h(s^*(m,f,g))\cdot s^*_2(m,f,g)\]
where $h$ is the equilibrium score density. Abusing notation, let $\rho_m:
\mathbb R\to\overline \Theta$ be the type with marginal cost $m$ that bids $s$ at the distribution $g$. Then we can rewrite $(\dagger)$ as
\[\int \int \left(\frac{\partial s^*(\rho_{m'}(s),g+\ep v)}{\partial\ep}\bigg\rvert_{\ep=0}-\frac{\partial s^*(\rho_m(s),g+\ep v)}{\partial\ep}\bigg\rvert_{\ep=0}~ dm'\right) \mathbb I_{s\le s^*(\theta,g)} h(s)~ ds.\]

From above, this expression must be negative for all $m$. Integrating $(\dagger)$ over all $m$ however gives us 
\[\int \int \int \left(\frac{\partial s^*(\rho_{m'}(s),g+\ep v)}{\partial\ep}\bigg\rvert_{\ep=0}-\frac{\partial s^*(\rho_m(s),g+\ep v)}{\partial\ep}\bigg\rvert_{\ep=0}~ dm'~dm\right) \mathbb I_{s\le s^*(\theta,g)} h(s)~ds=0.\]
This gives the desired contradiction. 
\end{enumerate}
\end{proof}

\textbf{Proposition \ref{prop:converse}.}\textit{
If a first-score auction has an informational precision of $n$ for $n<\infty$, then it satisfies equilibrium invariance.}

\begin{proof}
Assume there exists a type $\theta\in \text{Int}(\Theta)$ and two distributions $g_0$ and $g_1$ such that $[\theta]_{g_0}\not=[\theta]_{g_1}$. Then without loss of generality, let there be two types $(m,f_1)\in [\theta]_{g_0} $ and $(m,f_{n+2})\in [\theta]_{g_{n+2}}$ with $f_{n+2}<f_1$. Let $g_\lambda:=\lambda g_0+(1-\lambda g_1)$. By the differentiability of scores with respect to the distribution,  there exist $0=\lambda_1<\lambda_2<...<\lambda_{n+2}=1$ such that
\[(m,f_i)\in [\theta]_{g_i}\]
with $f_j<f_i$ if $j>i.$ By the continuity of the score bidding function, there exists a small open neighborhood $U_i$ around $(m,f_i)$ such that the set
\[S_{i}:=\bigg(U_i\cap \{\tau \in \Theta:\theta\succ_{g_{\lambda_i}} \tau\}\bigg)\subset \{\tau \in \Theta:\theta\succ_{g_{\lambda_i}}\tau
\text{ but }\theta\prec_{g_{\lambda_j}} \tau ~\text{ for all } j<i\}\]
and $S_i$ has positive measure for each $i\in \{1,...,n+2\}$. In words, the set $S_i$ contains the (competitive) types against which $\theta$ wins the auction at the distribution $g_{\lambda_i}$ but not at $g_{\lambda_j}$ for lower $j$.  Applying Lemma \ref{lem:moment_match} at each of these distributions, there must be types $\tau_i\in S_i$ and constants $a_i\not=b_i\in\mathbb R$ such that\footnote{See Figure \ref{fig:converse_b} for a graphical depiction of this.}
\[\psi_{\theta,g_{\lambda_i}}(\tau_i)=b_i \text{ and } \psi_{\theta,g_{\lambda_i}}(\tau_j)=a_i \text{ for all } j>i.\]
Aggregating this and recalling the definition of $\psi_{\theta,g_{\lambda_i}}(\tau)$,
\begin{adjustwidth}{-2cm}{-2cm}
\begin{equation*}\left(\begin{array}{ccc}
\zeta_\theta^1(\tau_1)  & \hdots &\zeta_\theta^n(\tau_{1})\\
\zeta_\theta^1(\tau_2)  & \hdots &\zeta_\theta^n(\tau_{2})\\
\vdots &  \ddots &\vdots\\
\zeta_\theta^1(\tau_{n+2})  & \hdots &\zeta_\theta^n(\tau_{n+2})\\
\end{array}\right)
\left(\begin{array}{ccc}
   \sigma_{\theta,1}(\theta,g_{\lambda_1}) & \hdots & \sigma_{\theta,1}(\theta,g_{\lambda_{n+2}})\\
   \sigma_{\theta,2}(\theta,g_{\lambda_1}) &  \hdots & \sigma_{\theta,2}(\theta,g_{\lambda_{n+2}})\\
   \vdots & \ddots &\vdots\\
   \sigma_{\theta,n}(\theta,g_{\lambda_1}) &  \hdots &\sigma_{\theta,n}(\theta,g_{\lambda_{n+2}})\\
\end{array}\right)
=
\left(\begin{array}{ccccc}
b_1 & \hdots &\hdots &\hdots&\hdots\\
a_1 & b_2 &\hdots & \hdots&\hdots\\
a_1 & a_2 & \hdots &\hdots\\
\vdots & \vdots & \ddots & \vdots&\vdots\\
a_1&a_2&\hdots & b_{n+1} &\hdots\\
a_1 & a_2 &... &a_{n+1}&b_{n+2}\\
\end{array}
\right)\end{equation*}\end{adjustwidth}

The rank of each matrix on the LHS is at most $n$ while the rank of the matrix on the RHS is at least $n+1$, giving a contradiction.
\end{proof}
    \color{black}

\section{Proofs from Section \ref{sub:ei_conds}}
In this section, I provide a proof of Theorem \ref{thm:rev_equiv} and Proposition \ref{prop:fsa_linear}.

\noindent \textbf{Theorem \ref{thm:rev_equiv}.}\textit{
A first-score auction with scoring rule $\Phi$ is informationally coarse if and only if it implements the same interim outcomes as the second-score auction with the same scoring rule at all distributions $g\in\mathcal G$.}
\begin{proof}
Fix a scoring rule $\Phi$ such that the first-score auction defined with this rule satisfies equilibrium invariance. Consider two types $\theta_1,\theta_2$ with the same break-even scores and assume for the sake of contradiction that $\theta_1\succ \theta_2$ in the first-score auction's equilibrium structure. Consider the sequence of distributions 
${(g_n)}_n\in \mathcal G$ defined by \[g_n=\frac{1}{n}g_0+\left(1-\frac{1}{n}\right)\cdot u\] 
where $u$ is the density of a uniform distribution over the set $\{\theta'\in \Theta:\theta_1\succ \theta'\succ \theta_2\}$.

Let $s_1$ be the score bid by $\theta_1$ at the distribution $u$ and let $\overline U$ be the payoff of type $\theta_2$ from bidding this score, conditional on winning. Note that $s_1<s_{BE}(\theta_1)=s_{BE}(\theta_2)$ so $\overline U$ is strictly positive. Note that for any type's equilibrium bid $(p,q)$, we must have that
\[0\ge p-(\bar m,\bar f)\cdot (q^\eta,1)\implies \bar m + \bar f \ge p-\theta\cdot (q^\eta,1)\]
for any $\theta\in \Theta$. Thus, conditional on winning, type $\theta_2$'s payoff is bounded above by $\overline m+\overline f.$

As $n$ grows, $X(\theta_2,g_n)\to 0$ but $X(\theta_1,g_n)\to 1$. Thus, let $N\in \mathbb R$ large enough such that 
\[\frac{X(\theta_1,g_N)}{X(\theta_2,g_N)}>\frac{\overline m+\overline f}{\overline U}.\] This implies that expected utility to type $\theta_2$ will be strictly greater from bidding $s_1$ than their bid in the mechanism. As a result, this mechanism is not incentive-compatible, giving a contradiction. 

This means that the first-score auction and the second-score auction have the same interim contract allocation rule $X$. In both auctions, all types in the equivalence class $[(\overline m,\overline f)]$ receive zero utility. By the envelope theorem, the slope of the indirect utility functions in both auction formats are the same upon changing fixed costs. Thus, fixing marginal costs and applying the Fundamental Theorem of Calculus on each of these lines yields that the interim utility functions from the two auction formats are the same. Then, also by the envelope theorem, the interim effort allocations are equal to the partial derivative of this utility function with respect to marginal cost, which are the same across these two mechanisms.
\end{proof}

\noindent \textbf{Proposition \ref{prop:fsa_linear}.}\textit{
A first-score auction with scoring rule $\Phi$ is informationally coarse if and only if the break-even effort function, $e_{BE}:={q_{BE}}^\eta:\Theta\to [0,1]$ is linear in fixed cost.}
\begin{proof}

    Fix such a first-score auction and a type $(m,f)\in \Theta$. Recall that given a score to bid, all types with the same marginal cost will optimize by picking the same contract. 
    
     As a result, if the other participant's type is in $[(m,z)]$ for $(m,z)\in\overline \Theta$, an agent with marginal cost $m$ will provide the contract $(p_{BE}(m,z),q_{BE}(m,z))$ in a second-score auction. Thus, the expected effort of this type in this auction, at any distribution $g\in\mathcal G$, is 
    \[\int e_{BE}(m,z) \cdot \mathbb I_{z\ge f} ~ g([(m,z)]) dz.\]

    From Proposition $1$, the expected effort of this type in a first-score auction is
    \[e_{BE}\left(m,\frac{\int z\cdot \mathbb I_{z\ge f}~ g([(m,z)]) dz}{\int \mathbb I_{z\ge f} ~g([(m,z)]) dz}\right)\cdot \int \mathbb I_{z\ge f} ~ g([(m,z)]) dz.\]

    These two expressions must be equal for a scoring rule that satisfies equilibrium invariance, by Theorem $2$. In a converse to Jensen's inequality, we'll show that these two expressions being equal for all $x$ implies that $e_{BE}$ must be linear. Firstly, we have that $e_{BE}(m,\cdot)$ is smooth because $\Phi$ is. This means its second-derivative must be continuous. Now assuming, $e_{BE}$ is not linear, there must be some subinterval of $[\underline f,\overline f]$ on which this function is strictly convex or concave. Let $[f_1,f_2]$ be the highest such interval and WLOG let the function be strictly convex here. Then we have that $e_{BE}$ is convex on $[f_1,\overline f]$ so by the strict form of Jensen's inequality,  
    \[e_{BE}\left(m,\frac{\int z\cdot \mathbb I_{z\ge f'}~ g([(m,z)]) dz}{\int \mathbb I_{z\ge f'} ~g([(m,z)]) dz}\right)< \frac{\int e_{BE}(m,z) \cdot I_{z\ge f'} ~ g([(m,z)]) dz} {\int I_{z\ge f'} ~ g([(m,z)]) dz}.\]

    for all $f'\in [f_1,f_2)$ giving us a contradiction. Thus, $e_{BE}$ must be linear in fixed cost.
\end{proof}

\section{Proofs from Section \ref{sec:microfound}}
\noindent \textbf{Lemma \ref{lem:belief_extend}.}
The first-order belief structure $\bm{\tilde \pi}$ uniquely defines beliefs of the form $\hat \pi_i:\mathcal T_i(I_i)\to \Delta(G\times \mathcal T_{-i}( I_{-i}))$ for all $i$.
\begin{proof}
By the definition of the state, each $g=(g_1,...,g_n)$ defines a unique distribution over opponents' payoff types, which is an element of $\Delta(\Theta_{-i})$. By the  definition of an information technology, for each agent $j$ and payoff type $\theta_j$, there is a unique element of $I_j(\theta_j)$ that contains $g$. In other words, there is a unique $t_j\in I_j(\theta_j)$ such that $g\in E(t_j)$ Thus, fixing a state $g$, which is a distribution over payoff types, we can extend this to a distribution over belief types as follows:
\[\mathbb P_g(A):=\mathbb P_g\bigg(\{\theta(t_{-i}):t_{-i}\in A, g\in E(t_{j}) \forall j\not=i\}\bigg)\]
for $A\subset \mathcal T_{-i}(\mathbb I_{-i})$.
We have uniquely mapped each element $g\in\mathcal G$ to a distribution over belief types. We can extend a distribution over states to a joint distribution over states and payoff types by interpreting the definition of each state in the support as a distribution over payoff types. We can then extend this joint distribution over states and payoff types to one over states and belief types with the definition above.

\end{proof}

\noindent \textbf{Theorem \ref{thm:rev_equiv}.}
\textit{A first-score auction with scoring rule $\Phi$ is informationally coarse if and only if it implements the same interim outcomes as the second-score auction with the same scoring rule at all distributions $g\in\mathcal G$.}
 \begin{proof}$(\Rightarrow)$ Starting with a coarse beliefs equilibrium, fix a type $t_i=(\theta_i,E_i)\in\mathcal T_i(I_i^*)$ and a distribution $g\in E_i$. For $i\in\{1,...,n\}$, consider a belief structure satisfying
        \[\tilde \pi_i(t_j)=\delta_g\]
        for all $t_j=(\theta_j,E_j)\in \mathcal T_j(I^*_j)$ such that $g\in E_j$. Under this belief structure, to type $t_i$, it will appear that the type distribution being $g$ is common knowledge. By the definition of a common prior equilibrium and robust preferences, this immediately implies that at this specific belief structure, each agent's action $a_i^*(t_i)$ is a best-response to other agents' equilibrium actions. Thus, this forms an equilibrium under the belief that $g$ is a common prior.\\
 $(\Leftarrow)$ Fix agent $i$ and assume that for all $g\in\mathcal G$, opponent $j$ with payoff type $\theta_{j}$ plays the strategy $a^{CP}_{j}(\theta_{j},g)$. I will show agent $i$ can do no better than adopting the common prior strategy. Let $t^*_i(g)$ be type  $\theta_i$'s belief type under the information acquisition strategy $I^*_i$ and let $a^*(t^*_i(g))=a_i^{CP}(\theta_i,g)$. This is well-defined because, by the given, 
 \[t_i^*(g)=t_i^*(g')\implies a_i^{CP}(\theta_i,g)=a_i^{CP}(\theta_i,g').\]
 We have that by definition,
        \[a_{CP}(\theta_i,g)\in \argmax_{a\in A}\mathbb E_{g}[U(\theta(t_i),\theta_{-i},a,a_{CP}(\theta_{-i},g)]\]
        Thus, for any prior belief $\pi$ and alternative strategy that yields the action $a(g)$ at the type distribution $g$, we have that 
        \[\mathbb E_{g\sim \pi}\bigg[\mathbb E_{g}[U(\theta(t_i),\theta_{-i},a(g),a_{CP}(\theta_{-i},g)]\bigg]\le \mathbb E_{g\sim \pi}\bigg[\mathbb E_{g}[U(\theta(t_i),\theta_{-i},a^*(t^*_i(g)),a_{CP}(\theta_{-i},g)]\bigg].\]
    The RHS is the expected utility of playing the strategy $(I^*,a^*).$ Because this holds for all $\pi,$ by definition this strategy is robustly optimal and thus this strategy profile forms a CBE. 
\end{proof}

\end{document}